\documentclass[aps,twocolumn,superscriptaddress,showpacs,amsfonts,amsmath,amsthm,amssymb,mathtools,latexsym,nofootinbib]{revtex4-1}
\usepackage[colorlinks,linkcolor=blue,urlcolor=blue,citecolor=blue,bookmarks,bookmarksnumbered]{hyperref}
\usepackage{breqn}
\usepackage{bm,braket,dcolumn}
\usepackage{slashed}
\usepackage{graphicx,subfigure}
\usepackage{epstopdf,epsfig}
\usepackage{enumitem}
\usepackage{amsthm}
\theoremstyle{definition}
\newtheorem{lemma}{Lemma}
\newtheorem{proposition}{Proposition}[section]
\newtheorem{property}{Property}[section]

\graphicspath{{pics/}{}}
\newcommand{\norm}[1]{\left\lVert #1 \right\rVert}
\newcommand{\modtwo}{~(\mathrm{mod}~2)}

\usepackage[usenames,dvipsnames]{color}         
\definecolor{purple}{rgb}{0.5,0,0.5}

\newcommand{\dist}[1]{\operatorname{dist}(#1)}
\newcommand{\vx}{\mathbf{x}}
\newcommand{\vy}{\mathbf{y}}
\newcommand{\vz}{\mathbf{z}}
\newcommand{\vi}{\mathbf{i}}
\newcommand{\vj}{\mathbf{j}}
\newcommand{\vA}{\mathbf{A}}
\renewcommand{\vr}{\mathbf{r}}
\newcommand{\ztwo}{\mathbb{Z}_2}
\newcommand{\bQ}{\overline{Q}}
\newcommand{\bS}{\overline{S}}
\newcommand\defeq{\mathrel{\stackrel{\makebox[0pt]{\mbox{\normalfont\tiny def}}}{=}}}
\newcommand{\tr}{\operatorname{Tr}}
\newcommand{\trv}{\tr ^v}
\newcommand{\refeq}[1]{Eq.~\eqref{#1}}
\newcommand{\Ker}{\operatorname{Ker}}
\newcommand{\Ind}{\operatorname{Ind}}

\begin{document}
\newcommand*{\PITT}{Department of Physics and Astronomy, University of Pittsburgh, Pittsburgh, Pennsylvania 15260, United States}
\newcommand*{\PQI}{Pittsburgh Quantum Institute, Pittsburgh, Pennsylvania 15260, United States}
\title{Local formula for the \texorpdfstring{$\ztwo$}{Z2} invariant of topological insulators}

\author{Zhi Li}\affiliation{\PITT}\affiliation{\PQI}
\author{Roger S. K. Mong}\affiliation{\PITT}\affiliation{\PQI}

\begin{abstract}
	We proposed a formula for the $\mathbb{Z}_2$ invariant for topological insulators, which remains valid without translational invariance. Our formula is a local expression, in the sense that the contributions mainly come from quantities near a point. Using almost commute matrices, we proposed a method to approximate this invariant with local information. The validity of the formula and the approximation method is proved.
\end{abstract}
\maketitle

\section{Introduction}
One of the most important progresses of condensed matter physics in recent years is the realization of many topological phases of matter beyond the Landau-Ginzburg paradigm. While the general classification of topological phases is still in progress, the classification for gapped non-interacting fermions is well-established \cite{AZ,Shinsei,KitaevK,twistK} and shows beautiful connections to $K$-theory and symmetric spaces. According to the action of several discrete symmetries, systems are classified into 10 classes. In each class, systems are labeled by a topological invariant valued in $\mathbb{Z}$ or $\ztwo$. The pattern appearing for various dimensions can be naturally explained by the Bott periodicity~\cite{Bott1,Bott2} and can be arranged into a periodic table.

A topological insulator, first proposed by Kane and Mele in Ref.~\onlinecite{Kane-Mele}, is a nontrivial system in two-dimension (2D) with time-reversal symmetry which squared to $-1$ (AII class in the Altland-Zirnbauer classification~\cite{AZ}). It is characterized by the gapless helical edge modes protected by the time-reversal symmetry \cite{Kane-Mele-edge}, and the band-crossing in the language of topological band theory. The topological invariant in this case is a $\ztwo$ number which we call Kane-Mele invariant. 

For systems with translational invariance, one can get analytical formulas for the topological invariants by working in momentum space and  considering essentially some vector bundles (with symmetries) \cite{Shinsei,KitaevK,twistK} over the Brillouin zone. For example, see Refs.~\onlinecite{Kane-Mele,Fu-Kane,Fu-Kane-inversion,Moore,Qi-Zhang,Roy} for various formulas for the $\ztwo$ Kane-Mele invariant.

While the classification is believed to be robust against disorder \cite{Cstar,noncommu,KitaevK}, analytical formulas are more difficult to find. Nevertheless, one can still get useful results from noncommutative geometry/topology considerations~\cite{Connes,Prodan_2010,Prodan_2011,Prodan_2013}, which may manifest itself as a (Fredholm, mod 2 Fredholm, Bott, etc) index~\cite{Zindex,Loring1,Loring2,Loring3,Loring4,SB2015,index1,Akagi_kernel,PhysRevB.84.241106,HuangLiu} (although some of them are abstract definitions and do not tell us how to calculate them efficiently); or from physical considerations such as scattering theory~\cite{Akhmerov}. A nice example is the following formula~\cite{KitaevAnyon,Mitchell2018} for two-dimensional Chern insulators (class A):
\begin{equation}\label{eq-Kitaev}
    \nu(P)
    =12\pi i\sum_{j\in \mathrm{A}}\sum_{k\in \mathrm{B}}\sum_{l\in \mathrm{C}}
    (P_{jk}P_{kl}P_{lj}-P_{jl}P_{lk}P_{kj}),
\end{equation}
where $P$ is the orthogonal projection operator onto filled bands, or equivalently the ground-state correlation matrix. \{A, B, C\} is a partition of the plane into three parts, as in Fig.~\ref{pic-3regions}. This formula reveals the local nature of the Chern number: assuming $P_{ij}$ decays fast enough as $|i-j|\to\infty$, then $v(P)$ can be well approximated by only summing over $j,k,l$ near the intersection point. For example, truncate the plane with a circle as in Fig.~\ref{pic-3regiontrunc}, then the same summation (with A, B, and C now finite) provides a good estimation. 
 \begin{figure}[htp]
	\centering
	\subfigure[]{
		\includegraphics[width=0.3\columnwidth]{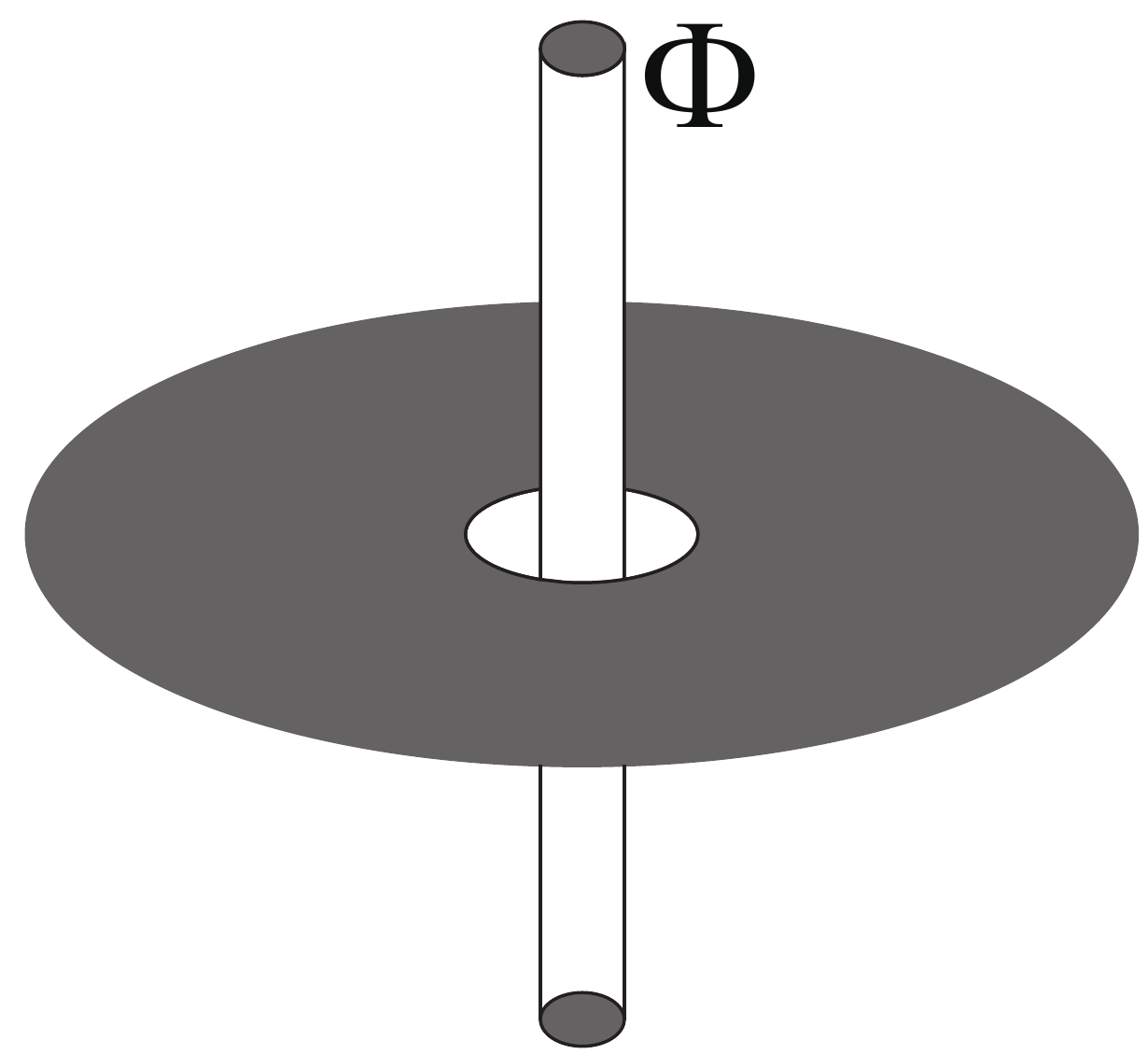}
		\label{pic-insert} }
	\subfigure[]{
		\includegraphics[width=0.25\columnwidth]{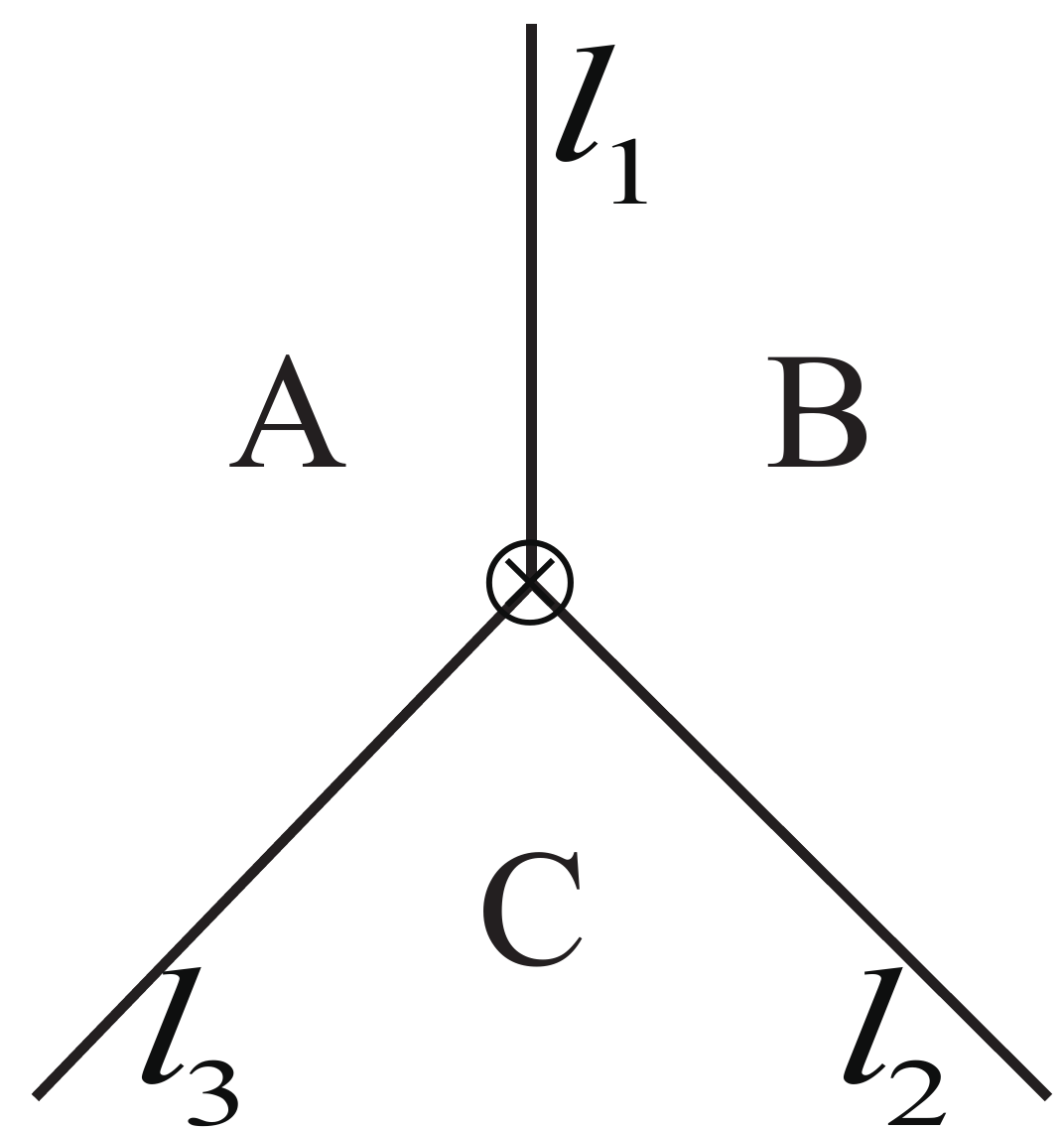}
		\label{pic-3regions} }
	\subfigure[]{
		\label{pic-3regiontrunc} 
		\includegraphics[width=0.28\columnwidth]{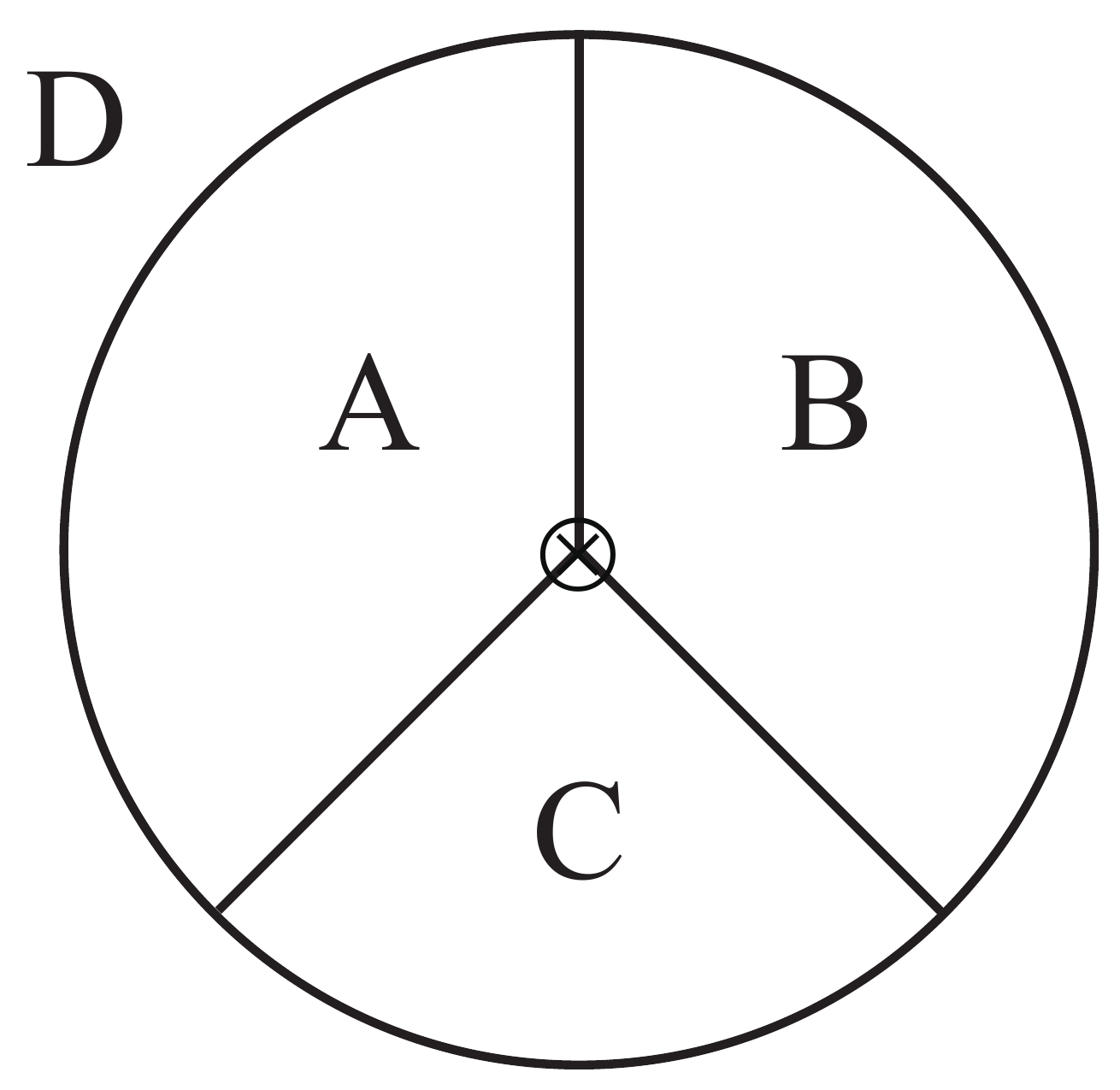}}
	\caption{(a) Insert a flux in the hole. (b) Divide the plane into three regions A, B, C. The intersection point is where a flux will be inserted. (c) Truncate the plane with a circle. Denote the region outside the circle by D. The intersection point is where a flux will be inserted.}
\end{figure}

In this paper, we propose a formula for the $\ztwo$ invariant for topological insulators in two dimensions, which remains valid with disorder.
Importantly, our approach is \emph{purely topological}, in the sense that we discard many geometrical information/choices such as distances and angles [see \refeq{eq-topQ}].
Moreover, we only require a mobility gap instead of a spectral gap. Similar to \refeq{eq-Kitaev}, the input of our formula is the projection $P$. Also similar to \refeq{eq-Kitaev}, our formula is essentially a local expression, in the sense that the contribution mainly comes from quantities near a point. As a result, one can expect to calculate it with sufficient precision by a truncation near that point.

This paper is organized as follows.
In Sec.~\ref{sec-intuition}, we explain the physics intuition and give a physical derivation of our formula. In Sec.~\ref{sec-formula}, we formally state our formula and show that it is well-defined.
In Sec.~\ref{sec-finite}, based on the theory of almost commuting matrices, we introduce a method to numerically calculate the invariant from a finite-size system.
We present some numerical results in Sec.~\ref{sec-numerical}.
In Sec.~\ref{sec-property}, we investigate the properties of our formula and sketch the proof of our main proposition.
To keep the paper more accessible, some technical details are gathered into the Appendix~\ref{sec-app1}.

\section{Intuition--Flux insertion and topological invariant}\label{sec-intuition}
In this section, we put Chern insulator/topological insulator on a punctured plane and insert fluxes at the origin [see Fig.~\ref{pic-insert}]. We will explain how the physics of flux insertion is related to topological invariant. This section aims to explain our intuition and provide a physical derivation of our formula, hence, some statements here may not very rigorous. We will establish our results carefully in the following sections. 

Recall the simple case, Chern insulators, which can be realized in integer anomalous quantum Hall systems. In this case, we have the well-known Thouless charge pump \cite{ThoulessPump}: when a flux unit is adiabatically inserted, it induces an annular electric field, which in turn produces a radial electric current due to the Hall effect. As a result, electrons are pushed away from (or close to, depending on the sign of the current) the origin for ``one unit". In Fig.~\ref{pic-Chernband} we draw the band structure for boundary states (near the puncture). Diagrammatically, when a flux unit is inserted, every occupied state moves toward top right to a lower level. 
\begin{figure}[htp]
	\centering
	\includegraphics[width=\columnwidth]{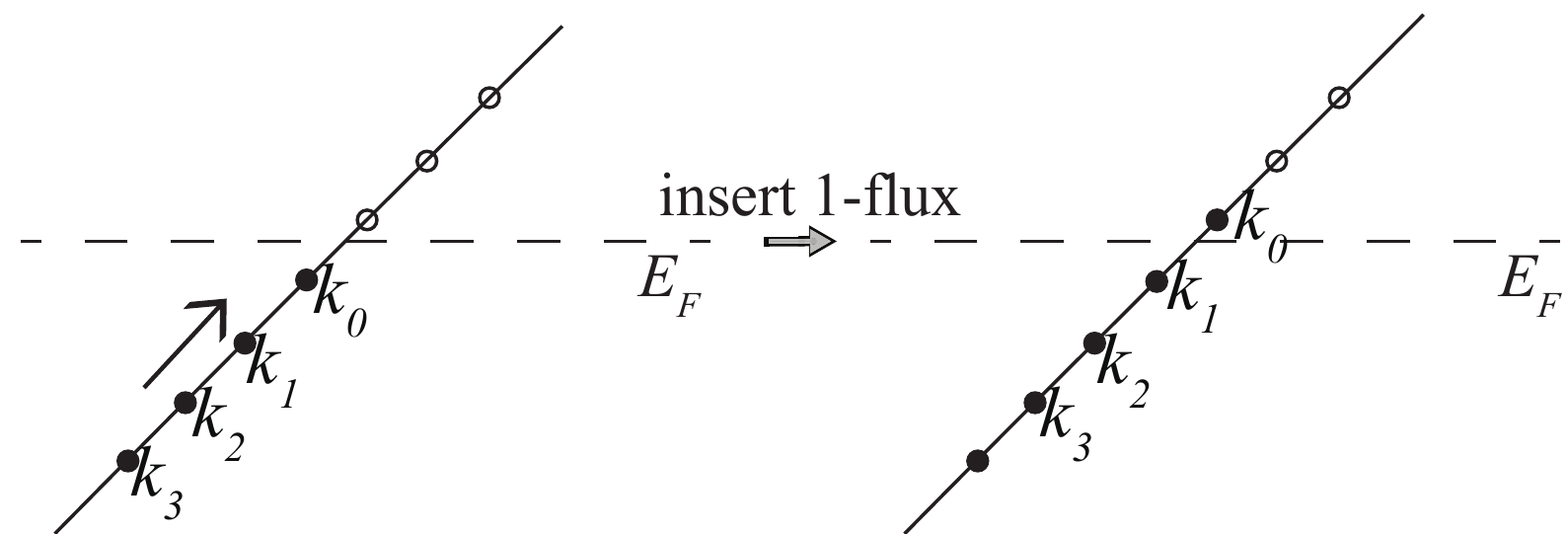}
	\caption{Band for boundary states of a Chern insulator. $\bullet$ means filled, $\circ$ means empty. After a unit flux insertion, every filled state moves towards top right to the next level. In this process, the label $k_i$ is tight to the electron, not the level.}\label{pic-Chernband}
\end{figure}

The many-body state after the flux insertion is not the ground state, because there is an filled state above the Fermi level. Compared to the ground states, we can see that the new ground state has one less electron ($k_0$ electron in Fig.~\ref{pic-Chernband}) than the old one (note that we are doing $\infty-\infty$, see comments below). The difference of number of electrons in ground states is exactly the Chern number. This is the idea behind Ref.~\onlinecite{Zindex}: 
\begin{align}\label{eq-relind}
&\text{Chern number}=\text{Ind}(P,P')\notag\\
=&\dim\Ker(P-P'-1)-\dim\Ker(P-P'+1),
\end{align}
where $P/P'$ is the projection operator onto filled states before/after the flux insertion, Ind is the relative index for a pair of projections, which intuitively counts the difference of their ranks (dimension of eigenvalue 1 subspace, number of filled levels in physics). Since the rank is just $\tr(P)$ and $\tr(P')$, one may expect
\begin{equation}
\text{Ind}(P,P')\sim\tr(P-P').
\end{equation}
This formula is indeed correct if $\tr(P-P')$ is well-defined---if $(P-P')$ is trace class \cite{Lax}. This is not the case for nontrivial Chern insulator though: $\Ind(P,P')$ is still well-defined \cite{Zindex}, but one needs a more complicated formula to evaluate it, which is essentially \refeq{eq-Kitaev}.

Now we turn to topological insulators. In this case, we adiabatically insert a $\frac{1}{2}$-flux quanta. As shown in Fig.~\ref{pic-TIband}, what happens is: energies for left-movers increase, while energies for right-movers decrease. If we assume (without loss of generality) the Fermi level is right below an empty state, the ground state after the flux insertion will have one more electron than before. We want to count the number of extra electrons to determine the Kane-Mele invariant $\nu_{\text{KM}}\modtwo$.
\begin{figure}[htp]
	\centering
	\includegraphics[width=\columnwidth]{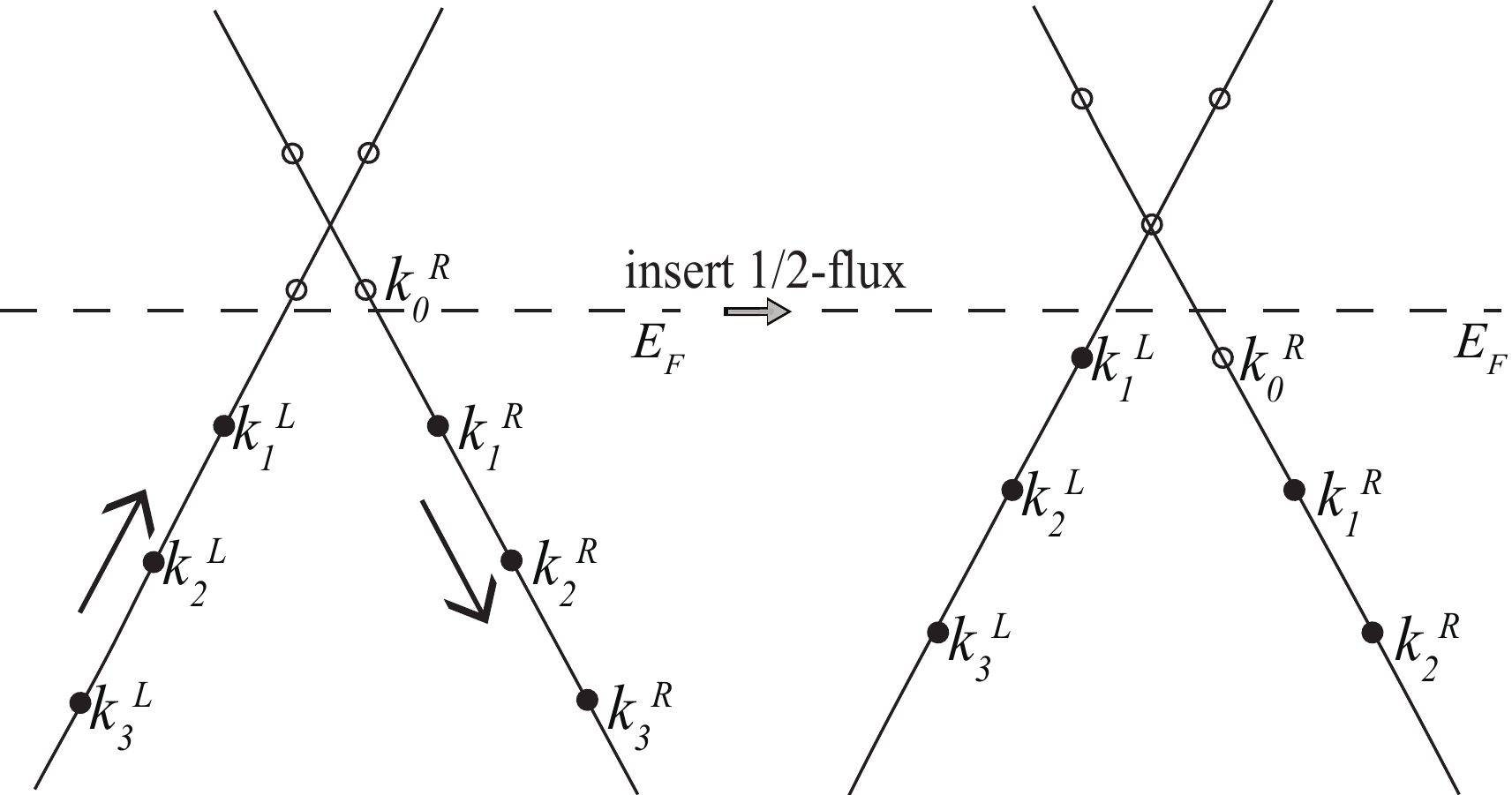}
	\caption{Band for boundary states of a topological insulator. After a half flux insertion, we get one more state under the Fermi level, which is geometrically near the vertex (flux).}\label{pic-TIband}
\end{figure}

To do this, we first count the number of electrons in a finite disk with radius $r$ (the system is still on an infinite plane, we just draw a virtual circle to define a disk). Due to time reversal symmetry, topological insulator have zero total Hall conductance, so the number of electrons inside the disk remains unchanged under adiabatic flux insertion. However, there is a vertex state ($k_0^R$ in Fig.~\ref{pic-TIband}) that is left empty, so in the new ground state the number of electrons in disk is increased by 1:
\begin{equation}\label{eq-dNdisk}
\Delta \braket{\text{No.~}\text{electrons in a (large) disk}}=1=\nu_{\text{KM}} \modtwo,
\end{equation}
where $\braket{}$ means ground state expectation value (note again that ground states before and after the flux insertion are different).

Since $P$ is the projection onto filled states, $H_0=1-P$ can be regarded as a spectral-flattened Hamiltonian (filled$=0$, empty$=1$). Denote $H_{\frac{1}{2}}$ to be the Hamiltonian after flux-insertion, consider $Q=1-H_{\frac{1}{2}}$ and corresponding projection matrix $\bQ$ (see Sec.~\ref{sec-formula} for details). We have
\begin{equation}
\begin{aligned}
&\braket{N_r}(\text{before})=\braket{\sum_{|i|<r}a_i^\dagger a_i}=\tr(P_r),\\
&\braket{N_r}(\text{after})=\tr(\bQ_r),\\
&\Delta \braket{N_r}=\tr{\bQ_r}-\tr{P_r}=\tr(\bQ_r-Q_r).
\end{aligned}
\end{equation}
Here $N_r$ is the number of electrons in disk $r$, $P_r,\bQ_r$ is the truncation of $P,\bQ$ ($\bQ_r$ means $(\bQ)_r$: spectral flatten before truncation). The last equation is because $P$ and $Q$ have the same diagonal elements (see Sec.~\ref{sec-formula}) and they are finite matrices.

Thus, we expect
\begin{dmath}\label{eq-guess}
	\nu_{\text{KM}}=\lim_{r\to\infty}\Delta \braket{N_r}=\lim_{r\to\infty}\tr(\bQ_r-Q_r)\sim\tr(\bQ-Q) \modtwo.
\end{dmath}

One may want to apply the same idea to Chern insulator. This will just lead to $0=0$. Indeed, we still have
\begin{equation}
\Delta \braket{\text{No.~}\text{electrons in a (large) disk}}=\lim_{r\to\infty}\tr(\bQ_r-Q_r).
\end{equation}
However, there will always be an electron go into (or out of) the disk adiabatically, which compensates the lost (or extra) state, so $\Delta \langle\text{No.~}\text{electrons in disk}\rangle$ in the left-hand side is always 0 in this case. This can also be seen from the (large) gauge equivalence between the two systems before and after a unit flux insertion. For the right-hand side, since $Q=P'$ in this case is already a projection, $\bQ=Q$, so the r.h.s is 0. The difference between topological insulators and Chern insulators is as follows: in the former case the number of electrons go through the boundary $r$ is 0 in average (because of zero Hall conductance), while in the latter it is nonzero and is essentially the Chern number.

As a side note, one may also consider the insertion of a unit flux and consider the difference between two ground states. A direct application of the relative index \refeq{eq-relind} gives 0. However, one may note that two terms (dimension of the kernel) in \refeq{eq-relind} come from left movers and right movers separately and one can therefore define the $\ztwo$ index with one kernel. This is the idea behind Ref.~\onlinecite{Akagi_kernel}.

\section{Formula for infinite system}\label{sec-formula}
In this section, we will carefully define the quantities in our main formula \refeq{eq-guess} and show its well-definedness.

The input of our formula will be the single-body projection operator $P$, which is related to the spectrum-flattened Hamiltonian $H_0=1-P$. For gapped states, $P$ decays at least exponentially \cite{Hastings_decay}: 
\begin{equation}
P_{\vx,\vy}<C_1e^{-C_2|\vx-\vy|}.
\end{equation} 
According to the Peierls substitution \cite{Peierls}, if we insert a 1/2-flux at a vertex, the new (single-body) Hamiltonian can be written as $H_{\frac{1}{2}}=1-Q$, where 
\begin{equation}
Q_{\vx,\vy}=s_{\vx,\vy}P_{\vx,\vy}.
\end{equation}
Here $s_{\vx,\vy}$ are phases such that for any loop $l=(\vx_1,\vx_2,\cdots,\vx_n=\vx_1)$, we have
\begin{equation}\label{eq-geoQ}
s_{l}\defeq\prod_{i=1}^n s_{\vx_i,\vx_{i+1}}=\begin{cases}
-1, & \text{if the vertex is in the loop} \\
1, & \text{otherwise}
\end{cases}.
\end{equation}
These phases can be chosen as follows: we divide the plane into three regions, as in Fig.~\ref{pic-3regions}. 

Let $n_{\vx,\vy}$ to be the number intersections of the straight line segment $(\vx,\vy)$ with three boundaries $l_1,l_2,l_3$, set
\begin{equation}\label{eq-phase}
s_{\vx,\vy}=(-1)^{n_{\vx,\vy}}.
\end{equation}
\newcommand{\primeup}{{}^{'}}
We call this gauge ``insert half fluxes along the boundaries". While $P$ is a projection, $Q$ no longer is. Actually, we have
\begin{dmath}
(Q^2-Q)_{\vx,\vy}=\sum_{\vz}s_{\vx,\vz}s_{\vz,\vy}P_{\vx,\vz}P_{\vz,\vy}-s_{\vx,\vy}P_{\vx,\vy}\\=-2s_{\vx,\vy}\sum_{\vz}\primeup P_{\vx,\vz}P_{\vz,\vy}.
\end{dmath}
where $\sum'$ means sum under constraint $s_{\vx\vy\vz}=-1$. Denote $V=Q^2-Q$. Since matrix elements of $P$  decays exponentially, $V$ is mainly supported around the vertex (hence the notation $V$) due to the constraint. To be specific, we have the following:
\begin{proposition}
$\exists C'_1,C'_2$, such that $|V_{\vx,\vy}|<C'_1e^{-C_2 r}$ where $r=\max\{|\vx|,|\vy|\}$.
\end{proposition}
\begin{proof}
Let us calculate $V_{\vx,\vy}$:
\begin{equation}
|V_{\vx,\vy}|=|2\sum_{\vz}\primeup P_{\vx,\vz}P_{\vz,\vy}|< 2C_1^2\sum_{\vz}\primeup e^{-C_2(|\vx-\vz|+|\vz-\vy|)}.
\end{equation}
From geometry, it is obvious that $|\vx-\vz|+|\vz-\vy|>r$ if $s_{\vx\vy\vz}=-1$, so the summation can be controlled by
\begin{equation}
2C_1^2e^{-\frac{C_2}{2}r}\sum_{\vz}e^{-\frac{C_2}{2}(|\vx-\vz|+|\vz-\vy|)}<C'_1e^{-C'_2 r}.
\end{equation}
(This is a pretty crude estimation but is enough.)
\end{proof}

In the following, we will refer call property as ``exponential decay property" (EDP). Intuitively, $Q^2-Q$ satisfies EDP means the deviation of $Q$ from a projection mainly comes from states near the vertex point. If we spectral flatten $Q$ to $\bQ$ (for eigenvalues $\lambda\leq\frac{1}{2}$, convert it to 0, otherwise convert it to 1), we anticipate that $Q-\bQ$ is mainly supported near the vertex. Actually $Q-\bQ$ also obeys EDP, but we do not need this result. We only need the following:
\begin{proposition}
	$Q-\bQ$ is trace class.
\end{proposition}
\begin{proof}
$|x-\bar{x}|\leq 2|x^2-x|$ for $\forall x\in\mathbb{R}$, so $|Q-\bQ|\leq 2|Q^2-Q|=2|V|$ as an operator (note that they commute). Since $V$ obeys EDP, $V$ must be trace class (see the corollary after Lemma \ref{lemma-est} in appendix \ref{sec-addproof}), so is $Q-\bQ$.
\end{proof}
Therefore, it is legal to define a ``trace over vertex states" as
\begin{equation}\label{eq-ourformula}
\trv(Q)=\tr(Q-\bQ).
\end{equation}
Note that in the definition of $\bQ$ we can arbitrarily choose the chemical potential $\mu\in(0,1)$, so the $\trv(Q)$ should be naturally understood as mod 1. In the case of topological insulator, $Q$ has time reversal symmetry, every states is Kramers paired, so $\trv(Q)$ can be naturally understood as mod 2. We will see in the following that it is $\trv(Q)$ mod 2 [instead of $\trv(Q)$ itself for a fixed ``chemical potential"] that has good properties. Also note that $Q$ is not trace class in general, so we cannot define $\trv(Q)$ as $\tr(Q)$ mod 2.

According to the above analysis, this expression is well-defined and the contributions mainly come from states near the vertex; it is a local expression. Interestingly, this local expression turns out to be independent of the flux-insertion point we choose; it only depends on the state itself. Moreover, it is an integer and is topologically invariant. 
Our main proposition is as follows:

\noindent
\textbf{Main Proposition} $\trv(Q)$ equals to the Kane-Mele invariant.

The derivation of our proposition is in Sec.~\ref{sec-property} and the appendix~\ref{sec-app1}. Before going on, we give three comments on our formula. 
\begin{enumerate}
    \item There is another construction of $Q$, closely related to the one given by \refeq{eq-geoQ}:
\begin{equation}\label{eq-topQ}
Q=\begin{bmatrix}
{AA}       & -{AB} & -{AC} \\
-{BA}       & {BB} & -{BC} \\
-{CA}       & -{CB} & {CC} 
\end{bmatrix},
\end{equation}
where ${AB}$ means $P_{AB}$, a block in the original matrix $P$.  If we consider a circle with many sites on it, it still gives us a total phase $-1$. In this case,  
\begin{equation}
Q^2-Q=2\begin{bmatrix}
0       & ACB & ABC \\
BCA       & 0 & CAB \\
CBA       & BAC & 0 
\end{bmatrix},
\end{equation}
where $ACB$ means $P_{AC}P_{CB}$, etc. It is still concentrated near the vertex (satisfies EDP), as long as the partition is good\footnote{For example, the one in Fig.~\ref{pic-3regions} is good. However, if we rotate $l_2$ towards $l_1$ and deform it a little bit so they are parallel at infinity, then $Q^2-Q$ does not satisfy EDP and convergence problem will occur.}. So we can follow the same procedure to define a new $\trv(Q)$.

The $Q$ defined here is \emph{not} unitary equivalent to the one in \refeq{eq-geoQ}---they have different spectra in general. However, in Sec.~\ref{sec-property} we will show that $\trv(Q)$ defined from them are equal (mod 2).  We call the $Q$ in \refeq{eq-topQ} the topological one, denoted by $Q^t$, because its definition does not depend on the geometric information such as ``straight line segments". The $Q$ in \refeq{eq-geoQ} will be called the geometric one, denoted by $Q^g$. It has the advantage of gauge invariance and many quantities [like spec$(Q)$] defined from it are manifestly independent of the partition. 

\item For systems in the DIII class ($\text{TRS}^2=-1$, $\text{PHS}^2=1$ where TRS is time-reversal symmetry and PHS is particle-hole symmetry), our formula can be simplified. 

Indeed, since the original system has PHS:
\begin{equation}
    K_{ph}^\dagger\overline{(2P-1)}K_{ph}=2P-1,
\end{equation} 
where $(2P-1)$ is the spectral-flattened Hamiltonian with spectrum=$\{\pm1\}$. $K_{ph}$ is an onsite action, and commutes with the operation from $P$ to $Q$, so the same equation holds for $Q$. Therefore the spectra of $Q$ is symmetric with respect to $1/2$:
\begin{equation}
\sigma(Q)=1-\sigma(Q).
\end{equation}
Now, for a spectrum $q$ such that $q\neq\frac{1}{2}$, the Kramers degeneracy and PHS provides us a four-fold $\{q,q,1-q,1-q\}$, which contributes 0 to $\trv(Q)\modtwo$. So
\begin{equation}\label{eq-d3}
\nu=\trv{Q}=\text{No.~}\{\text{Kramers pairs at }\tfrac{1}{2}\}\modtwo.
\end{equation}

\item The input $P$ is the correlation matrix for an infinite system. If we start with a finite system, say a topological insulator on the sphere, then our formula always gives 0.

Mathematically, this is because both $\tr(\bQ)$ and $\tr(Q)=\tr(P)$ are even due to time-reversal symmetry. Physically, it is because when inserting a flux at some point, it is unavoidable to insert another flux at somewhere else for a closed geometry, then our formula counts the vertex states at both points. To get the right invariant, we need to ``isolate" the physics at one vertex.
\end{enumerate}

\section{Approximation from finite system}\label{sec-finite}
Although the input of our formula is an infinite-dimensional operator $P$, our formula is a trace of vertex states, which should only depends on the physics near the origin. Let us truncate the plane with a circle $r$ [see Fig.~\ref{pic-3regiontrunc}]. Denote $P_N$ and $Q_N$ to be the truncation of $P$ and $Q$, where $N$ is the number of sites inside the circle, $N\sim r^2$. 
We expect that one can approximate the invariant with data near the origin, i.e.  with matrix elements of $P_N$ or equivalently $Q_N$. 

However, a naive limit $\lim_{N\to\infty}\tr(Q_N-\overline{Q_N})$ is wrong: it will give $\lim\tr(P_N)\modtwo$ since $\tr(\overline{Q_N})$ is even. This is because $\overline{Q_N}\neq(\overline{Q})_N$. Physically, $Q_N$ and $Q$ do have similar ``vertex states". However, unlike $Q$, $Q_N$ also includes boundary contributions [see \refeq{eq-QN2}], which need to be excluded. 

We claim that we can use the following algorithm to approximate our invariant.
\begin{itemize}
	\item Construct a matrix $V_N$ by
	\begin{equation}
		V_N=-2s_{\vx,\vy}\sum_{\substack{\vz\in(ABC)\\s_{\vx\vy\vz}=-1}}P_{\vx,\vz}P_{\vz,\vy}.
	\end{equation}
$V_N$ will be almost commuted with $Q_N$ and it will tell us whether a state is near the vertex or the boundary.
	\item Find approximations $Q'_N, V'_N$ for $Q_N, V_N$ so that they indeed commute\footnote{In practice, there are some arbitrariness to find $Q'_N, V'_N$. What we do is a joint approximation diagonalization (JAD) and then make them commute according to some rules. For example, one may make all $v'$ such that $|v'|>\epsilon$ to zero. Another rule is indicated in Sec.~\ref{sec-numerical}.}.
	\item Simultaneously diagonalize $Q'_N$ and $V'_N$ to get pairs of eigenvalues $(q',v')$. Sum over all the eigenvalues $q'$ such that $v'\neq 0$. 
	
	The summation will converge to $\trv(Q)$ as $N\to\infty$.
\end{itemize}

In the following we explain the algorithm in detail. First of all, we have:
\begin{dmath}\label{eq-QN2}
(Q_N^2-Q_N)_{\vx,\vy}=-s_{\vx,\vy}[2\sum_{\substack{\vz\in ABC\\s_{\vx\vy\vz}=-1}}P_{\vx,\vz}P_{\vz,\vy}+\sum_{\vz\in D}P_{\vx,\vz}P_{\vz,\vy}]
\defeq (V_N)_{\vx,\vy}+(W_N)_{\vx,\vy}.
\end{dmath}
Here, $V_N$ is supported near the center, while $W_N$ is supported near the boundary. This means the deviation of $Q_N$ to a projection happens both near the vertex and the boundary. 

We can also work in the topological construction of $Q$. In this case,
\begin{align}
Q_N^2-Q_N &= 2\left[\!\begin{smallmatrix}
0       & ACB & ABC \\
BCA       & 0 & CAB \\
CBA       & BAC & 0 
\end{smallmatrix}\!\right] - \left[\!
\begin{smallmatrix}
ADA     & ADB & ADC \\
BDA       & BDB & BDC \\
CDA       & CDB & CDC 
\end{smallmatrix}\!\right] \notag\\
&=V_N+W_N.
\end{align}
In both constructions, $Q_N, V_N, W_N$  should almost commute, and $V_N, W_N$ are almost orthogonal, since they are mainly supported in different regions (``almost" means relevant expressions approach 0 as $N\to\infty$).
\begin{proposition}
(1) $Q_N, V_N, W_N$ defined above satisfies
\begin{equation}
\begin{aligned}
\norm{[Q_N, V_N]}<\epsilon&,~\norm{[Q_N, W_N]}<\epsilon, \\
\norm{V_NW_N} <\epsilon&,~\norm{W_NV_N}<\epsilon,
\end{aligned}
\end{equation}
where the norm $\norm{\cdot}$ is the $L^2$ norm (maximal singular value), $\epsilon\sim p(r)e^{-r}$ where $p(r)$ is a polynomial of $r$. 

(2) There exist Hermitian matrices $Q'_N, V'_N, W'_N$ as approximations of $Q_N, V_N, W_N$ in the sense that
\begin{equation}
	\norm{Q_N-Q'_N}<\rho^2, ~\norm{V_N-V'_N}<\rho,~\norm{W_N-W'_N}<\rho,
\end{equation}
such that
\begin{align}
&[Q'_N, V'_N]=[Q'_N, W'_N]=V'_NW'_N=W'_NV'_N=0\\
&Q_N^{'2}-Q'_N=V'_N+W'_N.
\end{align}
Here $\rho$ can be chosen as $F(\epsilon)\epsilon^{1/10}$ (independent of $N$) where the function $F(x)$ grows slower than any power of $x$. 
\end{proposition}
\begin{proof}
	(1) Straight forward calculation.
	
	(2) It is easy to check $||Q_N||$ and $||V_N||$ are finite, independent of $N$ (one way to do this is to prove it for the topological construction $Q$ in \refeq{eq-topQ} and use the relationship between two constructions as in Property \ref{prop-topandgeo}). According to Lin's theorem \cite{Lin}, $\exists Q'_N, V'_N$ such that $\norm{Q_N-Q'_N}, \norm{V_N-V'_N}<\delta$ and $[Q'_N, V'_N]=0$. Moreover \cite{Hastings}, we can choose $\delta=E(1/\epsilon)\epsilon^{1/5}$ where the function $E(x)$ grows slower than any power of $x$, independent of $N$. 
	
	Define $W'_N={Q'}_N^2-Q'_N-V'_N$, then $W'_N,Q'_N,V'_N$ can be simultaneously diagonalized. Since $W_N=Q_N^2-Q_N-V_N$, $\norm{Q_N-Q'_N}, \norm{V_N-V'_N}<\delta$, so we have $\norm{W_N-W'_N}\lesssim\delta$ and 
	\begin{equation}
	\norm{V'_NW'_N}\!=\!\norm{(V'_N\!-\!V_N\!+\!V_N)(W'_N\!-\!W_N\!+\!W_N)}\!\lesssim\! \epsilon+\delta\!\sim\!\delta.
	\end{equation}
	This means for each pair of eigenvalues $(v'_N,w'_N)$, at least one of them should be smaller than $\sqrt{\delta}$. We manually make these eigenvalues to be 0, while fixing $v'_N+w'_N$. 
	
	The new $V'_N$ and $W'_N$ would be strictly orthogonal, and still commute with $Q'_N$, and still obeys $Q_N^{'2}-Q'_N=V'_N+W'_N$. Moreover, now $\norm{V_N-V'_N}\sim\delta+\sqrt{\delta}\sim\sqrt{\delta}\defeq\rho$. 
	
\end{proof}
Having $Q'_N, V'_N, W'_N$ exactly commute, and $V'_N, W'_N$ exactly orthogonal, we use them to distinguish vertex contributions and boundary contributions. We simultaneously diagonalize them and get triples $(q',v',w')$. Different contributions are then identified as follows (the reason for this identification is evident) [see Fig.~\ref{fig:numerics}]:
\begin{itemize}
	\item $v'\neq 0, w'=0$: vertex states
	\item $v'=0, w'\neq 0$: boundary states
	\item $q'=0~\text{or}~1,v'=w'=0$: bulk states
\end{itemize}
We anticipate that the summation of $q'$ over vertex states will be a good approximation of $\trv(Q)$.
\begin{proposition}[finite size approximation]\label{prop-finitesize}
	After the above procedure,
	\begin{equation}\label{eq-appro}
	\lim_{N\to\infty}\sum_{\substack{\text{vertex} \\ \text{states}}}q'=\trv(Q).
	\end{equation} 
\end{proposition}
The proof is in appendix \ref{sec-finiteproof}.

\begin{figure*}
    \begin{minipage}{32mm}(a)\\
        \includegraphics[width=32mm]{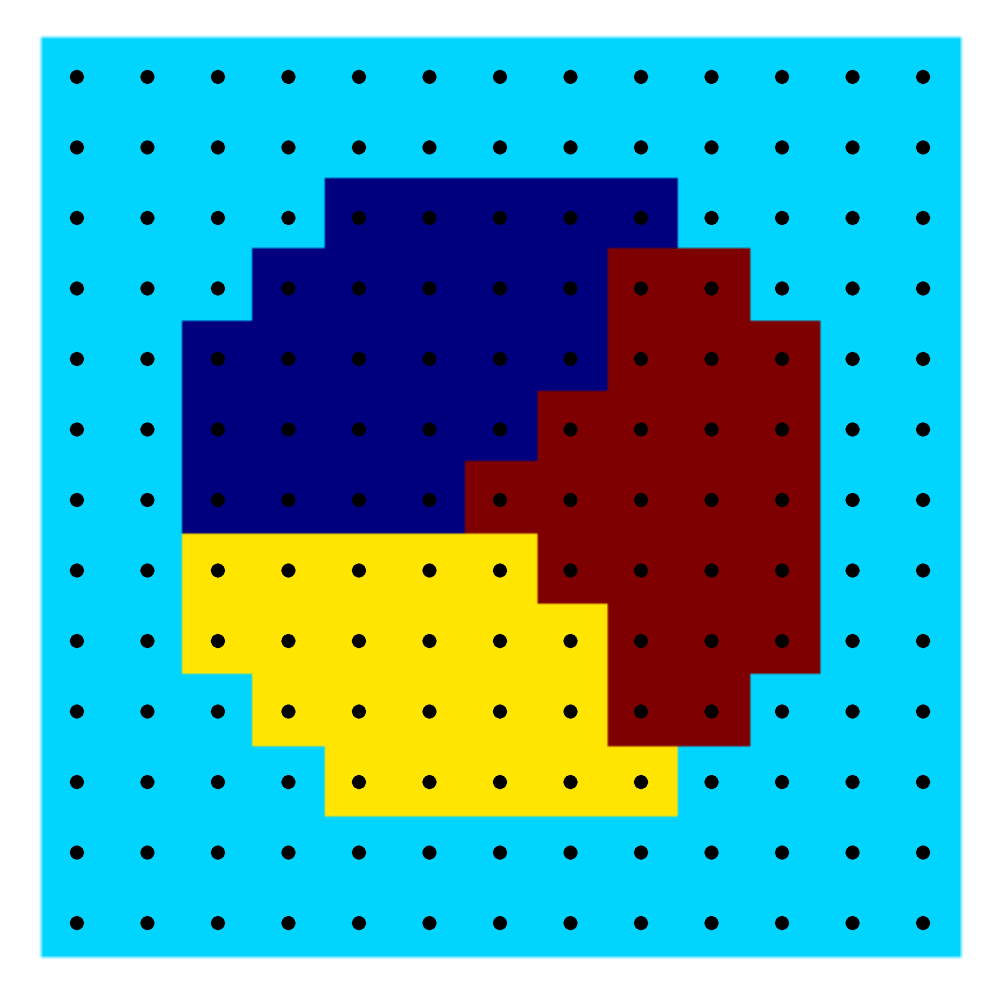}
    \end{minipage}
    \begin{minipage}{50mm}(b)\\
	    \includegraphics[width=50mm]{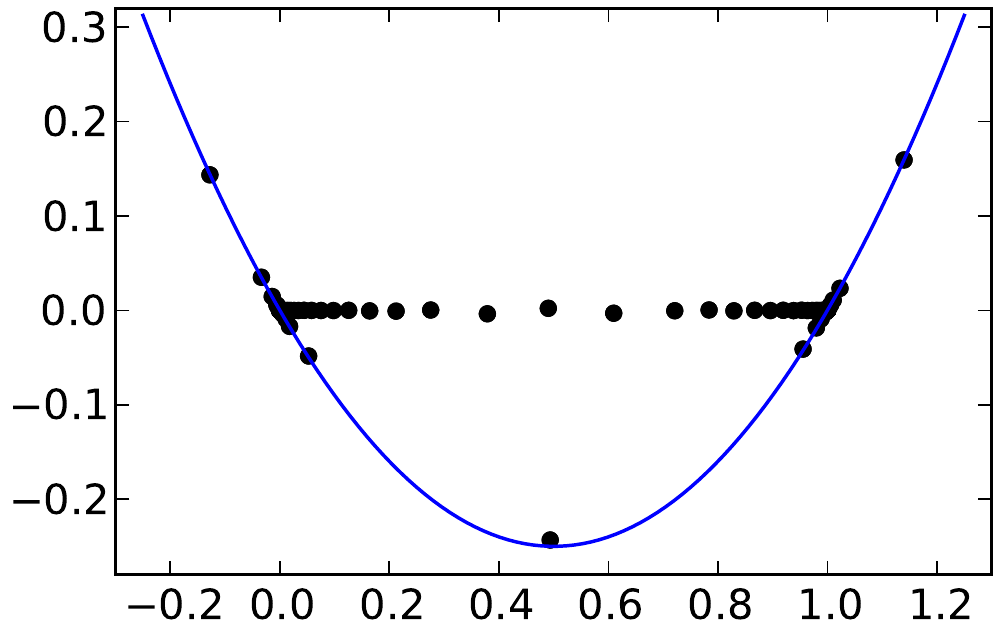}
	\end{minipage}
	\begin{minipage}{45mm}(c)\\
	    \includegraphics[width=45mm]{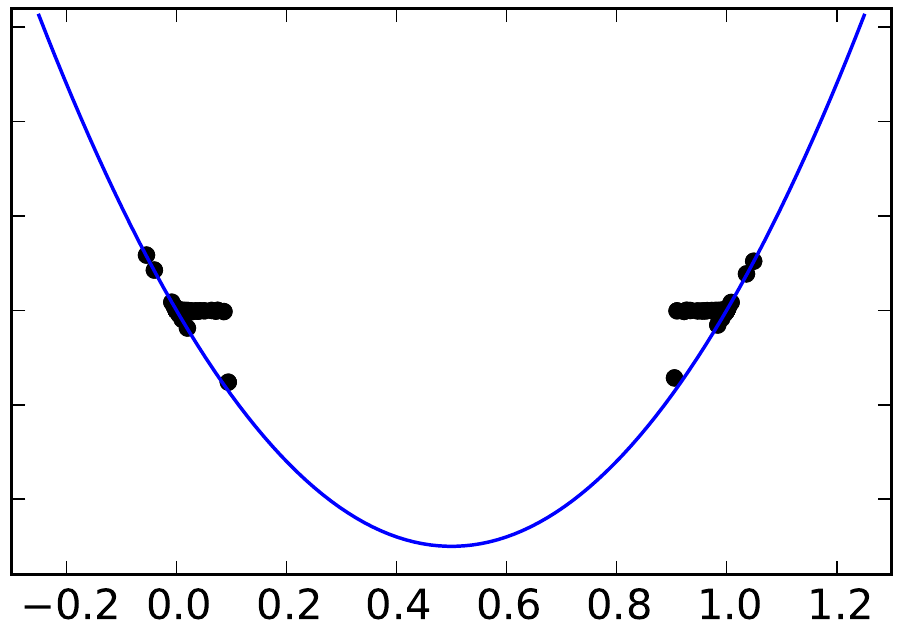}
	\end{minipage}
	\begin{minipage}{45mm}(d)\\
	    \includegraphics[width=45mm]{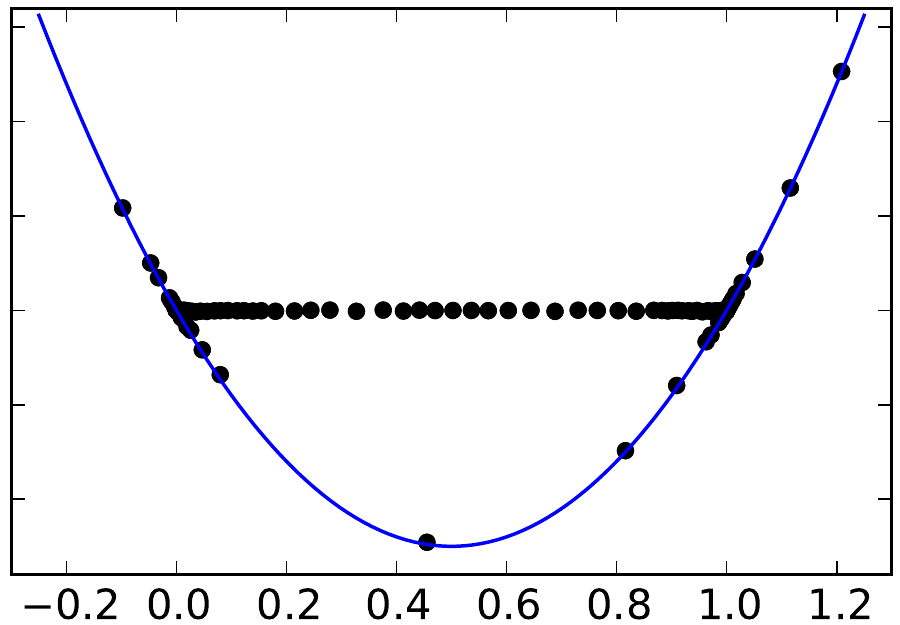}
	\end{minipage}
    \caption{Numerical results.
    (a) The geometry for computing the topological invariants.  The regions A, B, and C are represented by colours yellow, red, and navy blue respectively.  (Region D is represented by cyan.)
    (b)--(c) Numerical results for the Hamiltonian \refeq{eq-BHZlike} with (b) $m=1.6$ and (c) $m=2.4$.  Shown in the plot are eigenvalues of $V'_N$ vs.\ $Q'_N$.
    These results are generic; dots along the $x$-axis represent boundary states; dots near $(0,0)$ and $(1,0)$ are bulk states; dots on the parabola represent vertex states.
    (d) Numerical result from the Hamiltonian \refeq{eq-H3valley}, which strongly breaks particle-hole symmetry.
    }\label{fig:numerics}
\end{figure*}

\section{Numerical results}
\label{sec-numerical}
For our numerical results we use the Bernevig-Hughes-Zhang (BHZ) model~\cite{BHZ} on a square lattice with Rashba coupling and scalar/valley disorder:
\begin{align}\begin{split}
    H &= \int_{\mathbf{k}} (H_0 + H_R) ~d^2\mathbf{k}+ \sum_{\mathbf{r}} H_D(\mathbf{r}) ,
\\  H_0(\mathbf{k}) &= v(\tau^x\sigma^z\sin k_x+\tau^y\sin k_y) \\ &\quad+ (m-t\cos k_x-t\cos k_y)\tau^z ,
\\  H_R(\mathbf{k}) &= r(\sigma^x\sin k_y-\sigma^y\sin k_x) ,
\\  H_D(\mathbf{r}) &= V(\mathbf{r},+) \frac{1+\tau^z}{2} + V(\mathbf{r},-) \frac{1-\tau^z}{2} .
    \label{eq-BHZlike}
\end{split}\end{align}
Here there are four degrees of freedom per site, with $\tau$ and $\sigma$ acting on valley and spin space respectively.

In Figs.~\ref{fig:numerics}(b) and~\ref{fig:numerics}(c), we show the computation of the topological invariant of this model for $v = t = 1$, $r=1/2$, with the disorder $V(\mathbf{r},\pm)$ sampled uniformly from the interval $[-0.4,0.4]$.
Fig.~\ref{fig:numerics}(b) shows a topological phase with $m = 1.6$, while Fig.~\ref{fig:numerics}(c) shows a trivial phase with $m = 2.4$. 

These plots show the eigenvalues $(q',v')$ of the matrices $Q'_N$ and $V'_N$ respectively ($q'$ is along the $x$-axis and $v'$ is along $y$-axis).
Recall that we have $(q')^2-q'=v'+w'$ and $v'w'\approx0$, hence points $(q',v')$ either lives on the parabola $y=x^2-x$ (if $v'\neq 0$) or along the $x$-axis (if $v'=0$).
According to our analysis, points along the $x$-axis represent boundary states; points near $(0,0)$ and $(1,0)$ represent bulk states; all other points on the parabola corresponds to vertex states.
As a comparison, Fig.~\ref{fig:numerics}(c) shows the trivial region where there are (mostly) only have bulk states.
The goal of the numerical procedure outlined in the previous section is to isolate the vertex states which lives near the intersection of A, B, and C.

As the model \refeq{eq-BHZlike} (without disorder) is particle-hole symmetric, the resulting eigenvalues are reflection symmetric ($q' \to 1-q'$).  Disorder only breaks this symmetry very weakly.
To break this mirror symmetry, we construct a spinful model with three valleys,
\begin{align}\begin{split}
    H_0 &= 0.3\lambda^3 + 0.4\lambda^8
\\  &\quad + (0.5\lambda^1+0.6\lambda^4)\sigma^z\sin k_x
    + (1.2\lambda^2-0.3\lambda^5)\sin k_y
\\  &\quad + (0.5\lambda^1-0.7\lambda^3-0.3\lambda^4+1.1\lambda^6-0.6\lambda^8)\cos k_x
\\  &\quad + (0.5\lambda^1-0.7\lambda^3-0.4\lambda^4+1.0\lambda^6-0.6\lambda^8)\cos k_y,
    \label{eq-H3valley}
\end{split}\end{align}
where $\lambda^1,\dots,\lambda^8$ are the Gell-Mann matrices acting on valley space.
We retain the Rashba term with $r=0.1$, and disorder (for the three valleys independently) sampled from $[-0.4,0.4]$.  The spectrum $(q',v')$ is shown in Fig.~\ref{fig:numerics}(d), with $\nu$ evaluating to 1.02 indicating a QSH phase.



\begin{figure}
	\centering
	\includegraphics[width=0.9\columnwidth]{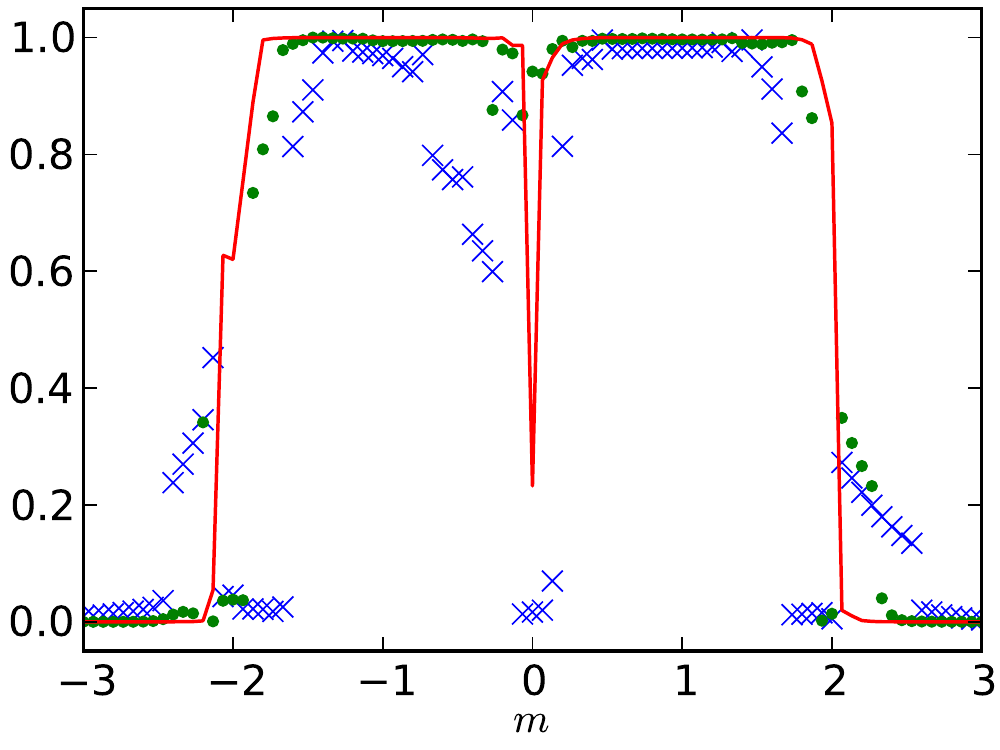}
	\begin{center}\begin{tabular}{cc@{\quad}c}
	& total system size & radius of ABC region \\ blue crosses & $7\times7$ & 2.4 \\ green dots & $13\times13$ & 4.5 \\ red line & $25\times25$ & 8.9
	\end{tabular}\end{center}
	\caption{The sum $\trv{Q}$ (approximating the invariant $\nu$) of a finite system for various $m$.
	The system sizes are shown in the table.  (For example, the green crosses shows data computed with Fig.~\ref{fig:numerics}(a)'s geometry.)
	}\label{pic-nu-m}
\end{figure}

In Fig.~\ref{pic-nu-m}, we plot the result of our formula \refeq{eq-appro} for model \refeq{eq-BHZlike} as a function of $m$.  (The data is computed for a single disorder realization.)
For the computation of $\trv{Q}$, we distinguished the vertex states (from the bulk and edge states) by only considering points satisfying $q' < 0$, $q' > 1$, or $v' < \frac{(q')^2-q'}{2}$.
We see that the system is in the Quantum spin Hall (QSH) phase for the bulk of $-2 \lesssim m \lesssim 2$. 
As the Hamiltonian $H_0 + H_R$ is gapless (Dirac-type) at $m=0$, we expect a small sliver of metallic phase near $m\approx0$ from disorder.
(In general, the metallic phase is stable in the AII class, hence we do not expect direct transition between the QSH and trivial phases.)
We see that as the system size increases (so that it is large compared with the correlation length $\xi$), the invariant approaches 0 or 1 to the gapped phases.

\section{Property and Proof}\label{sec-property}
In this section, we will investigate the property of $\trv(Q)$ and derive our main proposition step by step.
\begin{property}[gauge invariance]
Fixing the position of the flux and working in the geometric definition, then $\trv(Q)$ does not depend on the actual partition of the plane. For example, the following partition and the ordering of A, B, C will give the same $\trv(Q)$.
\end{property}

\begin{figure}[h!]
	\centering
	\includegraphics[width=0.5\columnwidth]{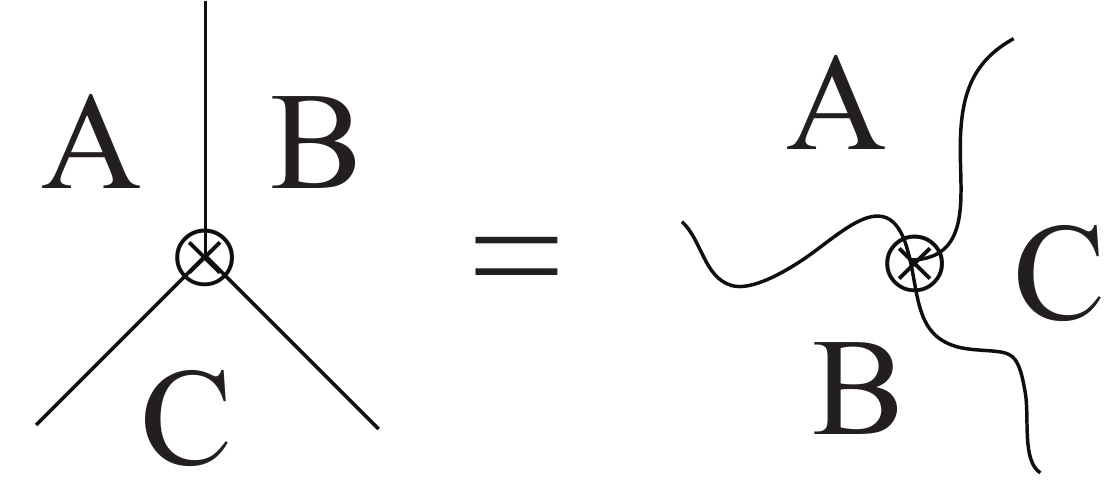}
\end{figure}
This is because different partitions correspond to different gauge choices in the Pierls substitution. Indeed, fix a reference point $\ast$, define $u_\vx=s_{\ast,\vx}/s'_{\ast,\vx}$ where $s,s'$ are the phases for two partitions. Since $s_{\ast,\vx}s_{\vx,\vy}s_{\vy,\ast}=s_{\ast\vx\vy}=s'_{\ast,\vx}s'_{\vx,\vy}s'_{\vy,\ast}$, we have:
\begin{equation}
s'_{\vx,\vy}=u_\vx s'_{\vx,\vy}\bar{u}_\vy
\end{equation}
Thus $Q'=UQU^\dagger$ and they have the same spectra.

\begin{property}\label{prop-topandgeo}
$\trv(Q^t)$ defined from topological $Q$ (for good partitions) and $\trv(Q^g)$ defined from geometric $Q$ are equal (mod 2).
\end{property} 
\begin{proof}
We calculate $Q^t-Q^g$  and find that
\begin{equation}
	(Q^t-Q^g)_{\vx,\vy}=\begin{cases}
		-2P_{\vx,\vy}, & \vx,\vy\text{~belong to different regions} \\
		&\text{and~}(\vx,\vy)\text{~intersects 2 boundaries} \\
		0, & \text{otherwise}
	\end{cases}.
\end{equation}
From geometry we can see if $|\vx|>r$, then the first condition is satisfied only if $|\vy-\vx|>\mathcal{O}(r)$ where $r=\max\{|\vx|,|\vy|\}$. So $Q^t-Q^g$ satisfies a EDP: $|(Q^t-Q^g)_{\vx,\vy}|<C_1e^{-C_2 r}$ and thus is trace class. Therefore,
\newcommand{\Qtb}{\overline{Q^t}}
\newcommand{\Qgb}{\overline{Q^g}}
\begin{itemize}
	\item $	\tr(Q^t-Q^g)=\lim_{N\to\infty}\tr(Q^t_N-Q^g_N)=0$, since they always have the same diagonal elements (note that we need trace class condition for the first equation evolving limit \cite{Lax}). 
	\item $\Qtb-\Qgb=(Q^t-\Qtb)-(Q^g-\Qgb)-(Q^t-Q^g)$ is also trace class.
\end{itemize}
Due to time-reversal symmetry,
\begin{equation}
\begin{aligned}
	&\tr(\Qtb-\Qgb)=\text{Ind}(\Qtb,\Qgb)\\
	=&\dim\Ker(\Qtb-\Qgb-1)-\dim\Ker(\Qtb-\Qgb+1)\\
	=&0 \modtwo,
\end{aligned}
\end{equation} 
where $\text{Ind}(\cdot,\cdot)$ is the index for a pair of projections \cite{Zindex}.
So we have
\begin{dmath}
	\trv(Q^t)-\trv(Q^g)=\tr(Q^t-\bar{Q^t}-Q^g-\bar{Q^g})=\tr(Q^t-Q^g)-\tr(\bar{Q^t}-\bar{Q^g})=0 \modtwo.
\end{dmath}
\end{proof}

Now we insert two 1/2-fluxes at different positions far away from each other. We divide the plane into four regions, as in Fig.~\ref{pic-4regions}. Again, we ``insert half fluxes along the boundaries" and write the Hamiltonian as
\begin{equation}
S_{\vx,\vy}=s_{\vx,\vy}P_{\vx,\vy},
\end{equation}
where $s_{\vx,\vy}$ are defined similar to \refeq{eq-phase} by the new partition. We have:
\begin{dmath}\label{eq-S2-S}
S^2-S=-2s_{\vx,\vy}\sum_{\substack{\vz\\s_{\vx\vy\vz}=-1}}P_{\vx,\vz}P_{\vz,\vy}=-2s_{\vx,\vy}(\sum_{\substack{\vz\\O_1\in(\vx\vy\vz)\\O_2\notin(\vx\vy\vz)}}+\sum_{\substack{\vz\\O_1\notin(\vx\vy\vz)\\O_2\in(\vx\vy\vz)}})P_{\vx,\vz}P_{\vz,\vy}\\
\defeq (V_1)_{\vx,\vy}+(V_2)_{\vx,\vy}.
\end{dmath}
Here $V_i~(i=1,2)$ is the ``vertex"  term for two junctions respectively. As in Sec.~\ref{sec-formula}, $V_i~(i=1,2)$ satisfies EDP for vertex $i$ and $S-\bS$ is trace class, so $\trv(S)$ is well-defined.

\begin{figure}[ht]
	\centering
	\includegraphics[width=0.5\columnwidth]{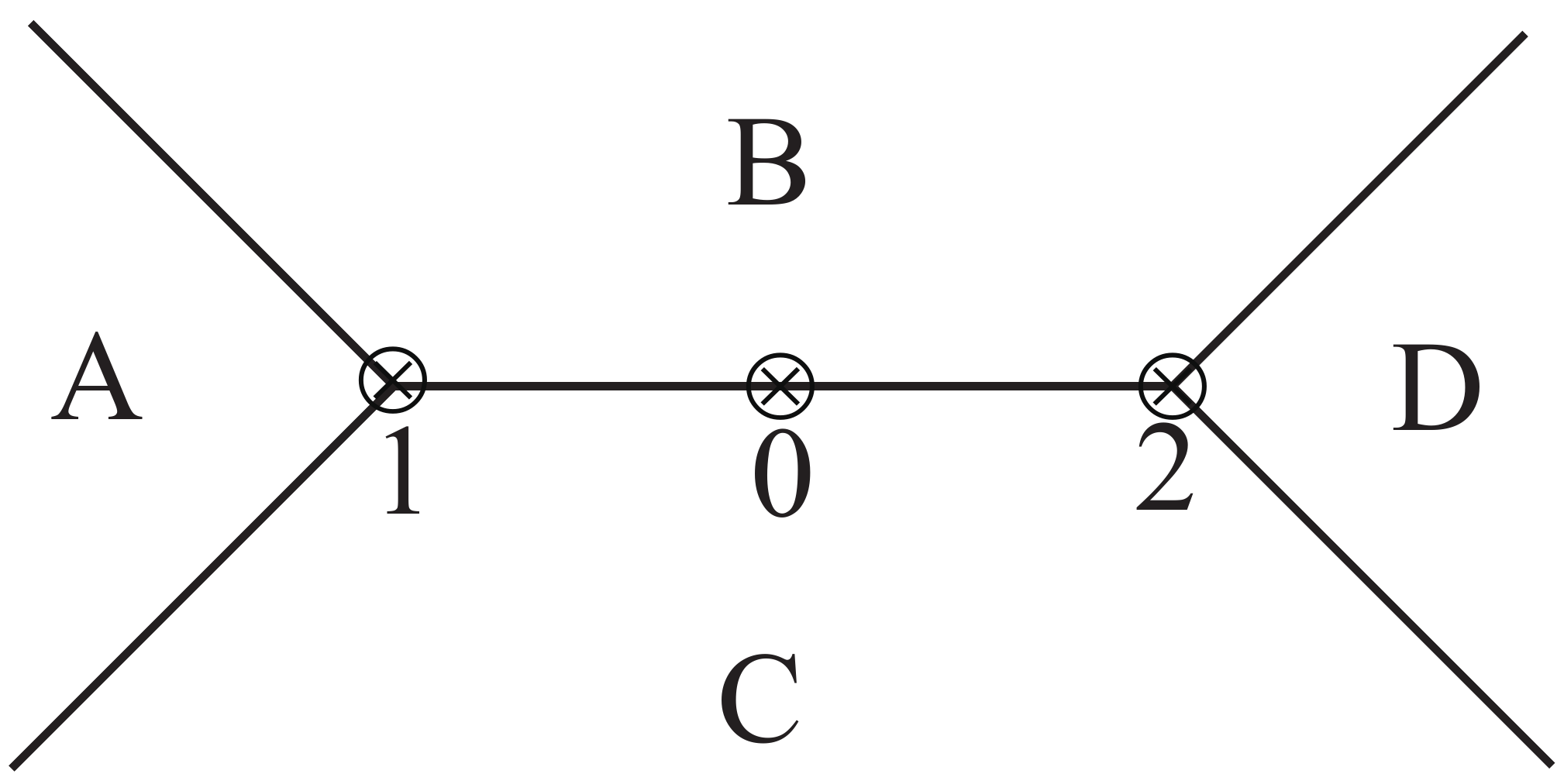}
	\caption{Divide the plane into 4 regions $A+B+C+D$. Insert a 1/2-flux for each vertex (1 and 2). We will show that it is ``equal" to insert a unit flux at the middle (point 0).}\label{pic-4regions}
\end{figure}

As $\text{dist}(1,2)\to\infty$, we have $V_1V_2=V_2V_1\to 0$. In the limiting case where $V_1V_2=V_2V_1=0$ exactly, one can simultaneously diagonalize them and at least one eigenvalue for a common eigenstate should be 0. This means each ``vertex state" of $S$ (those states with $S\neq 0,1$) is located at junction 1 or 2. Moreover, define $Q_1$ as the matrix corresponding to a $1/2$-flux insertion at point 1 with partition $A+B+CD$, $Q_2$ corresponding to the insertion at point 2 with partition $AB+C+D$, then the effect of $S$ for a state near junction $i$ will be close to the effect of $Q_i$, so one anticipates that 
\begin{equation}\label{eq-local1}
\trv(S)\approx\trv(Q_1)+\trv(Q_2).
\end{equation}

In the case where $\text{dist}(1,2)$ is large but not infinity, some vertex states of $S$ may come from the coupling of two vertex states at different fluxes. However, it is not difficult to convince that Eq.(\ref{eq-local1}) still holds. The physics here is similar to that for the two-states systems: due to the weak but nonzero coupling (off diagonal elements), the eigenstates will be approximately
\begin{equation}
\ket{\phi\pm}=\frac{1}{\sqrt{2}}(\ket{1}\pm\ket{2}).
\end{equation}
Here, we cannot say a vertex states of $R$ is located at one of the fluxes. However, the summation of eigenvalues for $\ket{\phi\pm}$ is still equal to that for $\ket{1}$ and $\ket{2}$. 
\begin{property}[additivity]\label{prop-add}
Under technical assumptions, as the distance between two fluxes goes to infinity, the vertex spectrum of $S$ ``decouples":
\begin{equation}\label{eq-add}
\lim_{\text{dist}(1,2)\to\infty}\trv(S)-(\trv(Q_1)+\trv(Q_2))=0.
\end{equation}
\end{property}
The proof is in appendix \ref{sec-addproof}.

\begin{property}\label{prop-S0}
 $\trv(S)=0 \mod 2$. This is an exact equation, no matter whether $\text{dist}(1,2)$ is large or not.
\end{property}
The idea is, if one looks from a large distance, inserting two 1/2-fluxes is approximately equivalent to insert a 1-flux, which is (singularly) gauge equivalent to 0-flux, so that no vertex states appear in the spectrum at all.   
\begin{proof}
We works in the AB gauge, where the vector potential of a flux satisfies 
\begin{equation}\label{eq-ABgauge}
|\mathbf{A}(\mathbf{r})|=\frac{\text{flux}}{2\pi r}, ~~~~\mathbf{A}(\mathbf{r})\perp \mathbf{r}.
\end{equation}
We still use $S$ to denote the Hamiltonian in this gauge. Denote $T$ to be the Hamiltonian for the case of inserting 1-flux at the center of two half-fluxes. $S$ and $T$ should be close outside the center. We will prove that $S-T$ is trace class in the Appendix \ref{sec-sttraceclass}. Then, similar to the proof of Property \ref{prop-topandgeo},
\begin{equation}
\tr(S-T)=\lim_{N\to\infty}\tr(S_N-T_N)=0,
\end{equation} 
since they have the same diagonal elements. $\bS-T=(\bS-S)+(S-T)$ is also trace class, so
\begin{dmath}
	\tr(\bS-T)=\text{Ind}(\bS,T)=\dim\Ker(\bS-T-1)-\dim\Ker(\bS-T+1)=0 \modtwo,
\end{dmath} 
due to time reversal symmetry. Therefore, 
\begin{equation}
\trv(S)=\tr(S-\bS)=\tr(S-T)-\tr(\bS-T)=0 \modtwo.
\end{equation}
\end{proof}

\begin{property}
$\trv(Q)$ is an integer$\modtwo$ independent of the position of the flux.
\end{property}
\begin{proof}
	For every pair of points 1 and 2, we have 
	\begin{equation}
	\trv(Q_1)-\trv(Q_2)=[\trv(Q_1)+\trv(Q_3)]-[\trv(Q_2)+\trv(Q_3)].
	\end{equation}
	Due to Properties~\ref{prop-add} and \ref{prop-S0}, it can be arbitrarily close to 0 (mod 2) as $\text{dist}(1,3)$ and $\text{dist}(2,3)$ goes to infinity. So we must have
	\begin{equation}
	\trv(Q_1)=\trv(Q_2) \modtwo.
	\end{equation}
Plug this into Eq.(\ref{eq-add}) and use again property \ref{prop-S0}, we obtain that  $\trv(Q_1)=\trv(Q_2)\in\ztwo$. 
\end{proof}

This property already shows that $\trv(Q)$ is an $\ztwo$ invariant for topological insulators which only depends on the state itself. The only natural identification is the Kane-Mele invariant.
\begin{property}
$\trv(Q)$ is equal to the Kane-Mele invariant.
\end{property}
\begin{proof}
Denote our invariant as $\nu$. For two gapped time-reversal-symmetric systems $A$ and $B$, we stack them (without hopping/interaction) and denote the new system $A\oplus B$. Obviously $\nu(A\oplus B)=\nu(A)+\nu(B)$. From the classification of topological insulator \cite{AZ,Shinsei,KitaevK} for AII class, the Kane-Mele invariant $\nu_{\text{KM}}$ is the only invariant with this property. It follows that
\begin{equation}
\nu=k\nu_{\text{KM}},
\end{equation}
where $k=0$ or $k=1$.

To prove $k=1$, it is enough to verify the existence of one system with $\nu=1$. To this end, consider a translational invariant system in the DIII class with nontrivial $\mathbb{Z}_2$ invariant. In this case, according to \refeq{eq-d3}, we have
\begin{equation}
\nu=\text{No.~}\{\text{Kramers pairs at }\frac{1}{2}\}=1\modtwo.
\end{equation}
The last equation can be obtained by considering the band structure, due to translational invariance: A Kramers pair at 1/2 correspond to the a band crossing.
\end{proof}

\section{Conclusion}
In this paper, we propose a formula \refeq{eq-ourformula} for the $\ztwo$ invariant for topological insulators in 2D, which remains valid with disorder. The intuition behind our formula is flux-insertion-induced spectral flow, which manifest itself as the difference of numbers of electrons in the ground states. The formula works by taking the single-body projection matrix $P$ (or ground state correlation function) as the input, performing a Peierls substitution (either geometrical or topological), and then take the ``trace over vertex states". Our formula is a local expression, in the sense that the contribution mainly comes from quantities near an arbitrarily but fixed point. The validity of this formula is proved indirectly, by showing its properties (gauge invariance, addictivity, integrality, etc). All properties are physically explained and mathematically proved.

Due to the local property of our formula, it can be well approximated with partial knowledge of the projection matrix. In this case, we construct ``vertex matrix" and ``boundary matrix" which almost commute. Using an interesting parabola, the vertex contributions are separated out. The validity of this algorithm is proved and verified numerically.

Similar ideas may be used in the case of other symmetries and other dimensions. For example, for class A in 2D (Chern insulator), an infinitesimal flux insertion will reproduce \refeq{eq-Kitaev}. It would be interesting to work out other cases to see if one can get (and prove) a new formula. last but not least, the idea may also be extended to some interacting cases. It would be interesting to explore such generalizations.

\section{Acknowledgement}
This work is supported from NSF Grant No. DMR-1848336. ZL is grateful to the PQI fellowship at University of Pittsburgh.


\appendix
\section{More technicalities in the proof}\label{sec-app1}
\subsection{Proof of additivity}\label{sec-addproof}
In this appendix, we will prove Property\ref{prop-add} in Sec.~\ref{sec-property}. The proof is a little bit technical, but the physics idea is simple: vertex states of $S$ comes from those of $Q_1$ and $Q_2$.
\begin{lemma}[Almost orthogonal vectors]\label{lemma-combination}
For $N$ unit vectors $u_n$ in a $d$-dimensional linear space s.t. $|(u_i,u_j)|<\sigma$ for each $i\neq j$. If $\sigma<\frac{1}{\sqrt{2d}}$, then $N<2d-1$. 
\end{lemma}
\begin{proof}
Let $A=(A_{i,j})=((u_i,u_j))$ to be the Gram matrix of $\{u_i\}$. Denote $\lambda_1,\cdots,\lambda_{N}$ the eigenvalues of $A$. Since $u_i\in\mathbb{C}^d$, $\text{rank}(A)\leq d$, at most $d$ of them are nonzero. By Cauchy inequality, $\sum \lambda_i^2\geq(\sum\lambda_i)^2/d=(\tr{A})^2/d=N^2/d$. On the other hand, $\sum \lambda_i^2=\tr(A^2)=\sum|(u_i,u_j)|^2<N+N(N-1)\sigma^2$. Therefore,
\begin{equation}\label{eq-combination}
\frac{N^2}{d}<N+N(N-1)\sigma^2<N+N(N-1)\frac{1}{2d}.
\end{equation}
Solving this inequality, we find $N<2d-1$.
\end{proof}
\begin{lemma}[an estimation of the eigenvalue distribution]\label{lemma-est}
	Assume a Hermitian matrix $A$ satisfying the exponential decay property (EDP) $|A_{\vi,\vj}|<C_1 p(t)e^{-C_2t}$ where $t=\max\{|\vi|,|\vj|\}$, $p(t)$ is a monic (leading coefficient=1) polynomial. The number of eigenvalues outside $(-\epsilon,\epsilon)$ is bounded by $\mathcal{O}(\frac{1}{C_2^2}\ln^2{\frac{C_1}{C_2^\alpha\epsilon}})$, where $\alpha=2+\deg p$.
\end{lemma}
\begin{proof}
For an unit eigenvector $x$: $Ax=ax$, $|a|>\epsilon$, we separate $x$ into two parts $x=y\oplus z$ according to whether the label is inside or outside a circle: $y_\vi=0$ if $|\vi|>r$, $z_\vi=0$ if $|\vi|<r$. The radius $r$ will be determined later (depend on $\epsilon$).

We claim that $\norm{z}<\delta\defeq\frac{C_1r^{3/2}p(r)e^{-C_2r}}{\epsilon}$. Indeed, $ay\oplus az=Ax$, $||x||=1$. According to the Cauchy inequality we have
\begin{equation}\label{eq-a1}
\norm{\epsilon z}^2<\norm{az}^2<\sum_{|\vi|>r,\vj}{|A_{\vi,\vj}|^2}\lesssim[C_1r^{\frac{3}{2}}p(r)e^{-C_2r}]^2.
\end{equation}
Here, $\lesssim$ (means inequality up to constant) can be verified by doing integral. 
Denote the number of eigenvalues larger than $\epsilon$ to be $N$: $Ax_n=a_nx_n$, $n=1,\cdots,N$. Without loss of generality, we can assume they are orthogonal, so
\begin{equation}
(x_i,x_j)=0\Rightarrow |(y_i,y_j)|=|(z_i,z_j)|<\delta^2.
\end{equation}
Now we have $N$ unit vectors $u_n=\frac{y_n}{\sqrt{1-z_n^2}}$ in dimension $d\propto r^2$ whose inner products with each other are less than $\sigma\defeq\frac{\delta^2}{1-\delta^2}$. We choose $r$ large enough so that 
\begin{equation}\label{eq-cond}
\sigma<\frac{1}{\sqrt{2d}}\sim\frac{1}{r}.
\end{equation}
According to Lemma \ref{lemma-combination}, $N<2d=\mathcal{O}(r^2)$.
We can solve \refeq{eq-cond} to estimate $r$ (thus $N$). Roughly, set $\sigma\sim(\frac{C_1r^{3/2}p(r)e^{-C_2r}}{\epsilon})^2=\frac{1}{r}$, let $x=C_2r$, we find $e^x=\frac{C_1}{C_2^\alpha\epsilon}x^\alpha$, where $\alpha=2+\deg p$. The exact solution can be expressed using the Lambert $W$ function \cite{LambertW}. Here we only need the asymptotic expansion. Denote $\beta=\frac{C_1}{C_2^\alpha\epsilon}$, take logarithm, we have the following iteration:	
\begin{dmath}
	x=\ln\beta+\alpha\ln x=\ln\beta+\alpha\ln(\ln\beta+\alpha\ln x)=\cdots=\mathcal{O}(\ln\beta).
\end{dmath}
Thus, it is enough to choose $r=\mathcal{O}(\frac{1}{C_2}\ln{\frac{C_1}{C_2^\alpha\epsilon}})$, and $N<2d=\mathcal{O}(r^2)=\mathcal{O}(\frac{1}{C_2^2}\ln^2{\frac{C_1}{C_2^\alpha\epsilon}})$.
\end{proof}

As a corollary, it follows that
\begin{dmath}\label{eq-EDPistraceclass}
\sum_{|a|<\epsilon}|a|=\sum_{|a|<\epsilon}\int_{0}^{\epsilon}\theta(|a|-x)dx=\int_{0}^{\epsilon}\sum_{|a|<\epsilon}\theta(|a|-x)dx\\<\int_{0}^{\epsilon} \frac{1}{C_2^2}\ln^2{\frac{C_1}{C_2^\alpha x}}dx,
\end{dmath}
which converges to 0 as $\epsilon\to 0$.  Therefore any EDP operator is trace class.

\begin{lemma}\label{lemma-approx}
Assume $A$ is Hermitian. If $\exists x\neq 0$ s.t. $\norm{(A-\lambda)x}<\norm{\epsilon x}$, then $A$ has an eigenvalue $a\in(\lambda-2\epsilon,\lambda+2\epsilon)$. Moreover, decompose $x=x^{\parallel}+x^{\perp}$ with respect to subspace $(\lambda-2\epsilon,\lambda+2\epsilon)$, then $\norm{x^\perp}<\frac{1}{2}$. If $A$ is of finite size $N\times N$, then an eigenvector $x_a$ of $A$ with eigenvalue $a\in(\lambda-2\epsilon,\lambda+2\epsilon)$ satisfies $|(x,y)|>\sqrt{\frac{3}{4N}}$.
\end{lemma}
\begin{proof}
Denote $B=A-\lambda$, then $\norm{Bx}<\norm{\epsilon x}$. Without loss of generality, assume $\norm{x}=1$. Let us diagonalize $B$, so that $B=\text{diag}\{b_1,\cdots,b_n\}$. Then we have
\begin{dmath}
	\epsilon^2>\sum{b_i^2|x_i|^2}=\sum_{|b_i|\geq 2\epsilon}b_i^2|x_i|^2+\sum_{|b_i|<2\epsilon}b_i^2|x_i|^2\\
	>4\epsilon^2\sum_{|b_i|\geq 2\epsilon}|x_i|^2+\sum_{|b_i|<2\epsilon}b_i^2|x_i|^2.
\end{dmath}
So we must have $\sum_{|b_i|\geq 2\epsilon}|x_i|^2<\frac{1}{4}$, thus $\norm{x^\perp}<\frac{1}{2}$  and $\exists i$ such that $|b_i|<2\epsilon$. If $N$ is finite, then from $\sum_{|b_i|< 2\epsilon}|x_i|^2<\frac{3}{4}$ we know $\exists i$ such that $|b_i|<2\epsilon, |x_i|>\sqrt{\frac{3}{4N}}$.
\end{proof}

Go back to the original proposition. We want to find a correspondence between vertex eigenvalues of $S$ and those of $Q_1$ and $Q_2$. To make the notation simple, in the following $e^{-r}$ means $C_1p(r)e^{-C_2}r$ and $C_1, C_2$ can change.

For each vertex state $x$ of $Q_1$: $Q_1x=q x$, $\norm{x}=1$ $(q\neq 0,1)$, separate $x$ as $x=y+z$ with respect to the disk $B(1,r)$. As in Eq.(\ref{eq-a1}), we still have $\norm{z}<\frac{e^{-r}}{|q^2-q|}$. It is not difficult to show that
\begin{dmath}
\norm{(S-q)x}=\norm{(S-Q_1)x}\leq \norm{(S-Q_1)y}+\norm{(S-Q_1)z} \\
\lesssim (1+\frac{1}{|q^2-q|})e^{-r}\defeq \delta.
\end{dmath}
According to Lemma \ref{lemma-approx}, $S$ has an eigenvalue in $(q-2\delta,q+2\delta)$. The same argument also applies to $Q_2$. Thus, for each vertex eigenvalue $q$ of $Q_1,Q_2$, we get an eigenvalue of $S$ in a neighbourhood of $q$.

Denote $U_\epsilon\defeq(-\epsilon,\epsilon)\cup(1-\epsilon,1+\epsilon)$. For $q\notin U_\epsilon$, $|q^2-q|>\epsilon/2$, so that $\delta<\frac{e^{-r}}{\epsilon}$. As $r\to\infty$, we will adjust $\epsilon$ accordingly so that $\delta$ is still small enough such that spectral gaps outside $U_\epsilon$ are always greater than $\delta$. Then we can describe the spectra structure of $Q_1, Q_2$ in Fig.~\ref{pic-spectrum}. The shaded windows have width$\sim\delta$ and are of three types. For types 1 and 2, we already get the correspondence. For type 3, we claim the dimension of the subspace $X$ corresponding to eigenvalues (of $S$) in such window is at least 4. Indeed, denote $x_1,x_2$ the eigenstates of $Q_1$, $x_3,x_4$ the the eigenstates of $Q_2$, then $|(x_i,x_j)|<\delta$. Let $x_i=u_i+v_i$ where $u_i\in X$ and $v_i\perp X$, then 
\begin{equation}\label{eq-121}
|(y_i,y_j)|\leq|(x_i,x_j)|+|(z_i,z_j)|<\delta+\frac{1}{4}.
\end{equation}
Similar to \refeq{eq-combination} (here $N=4$), we get $d\geq 4$. 
\begin{figure}
	\centering
	\includegraphics[width=\columnwidth]{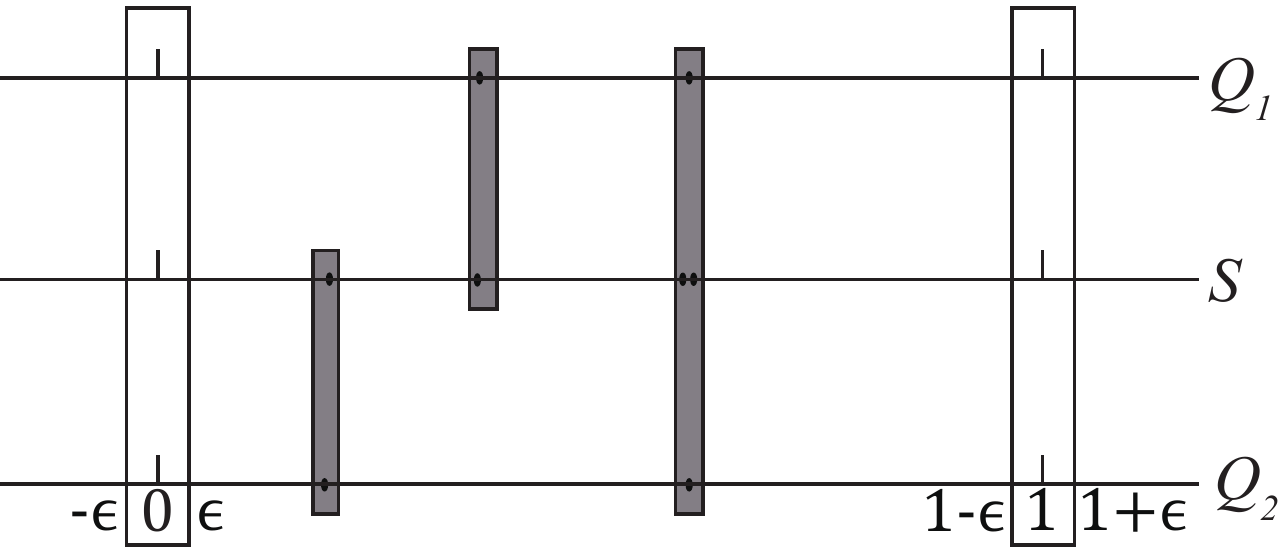}\\
	\caption{Spectra structure of $Q_i$ and $S$. We only consider eigenvalues outside $(-\epsilon,\epsilon)\cup(1-\epsilon,1+\epsilon)$. The dots $\bullet$ are eigenvalues. Each dots represents a Kramers pair due to time reversal symmetry. The shaded windows are of width $\delta\sim\frac{e^{-r}}{\epsilon}$ and are (from left to right) of three types.}\label{pic-spectrum}
\end{figure}

Moreover, each eigenstate of $S$ with $s\in U_\epsilon$ is generated in this way. Indeed, assume $Sx=sx$, $\norm{x}=1$ $(s\neq 0,1)$, then 
\begin{equation}
x=\frac{1}{s^2-s}(V_1x+V_2x)
\end{equation}
is a decomposition of $x$. At least one of $\norm{V_ix}$ should be larger than $|s^2-s|/2$, say $V_1x$. Note that $V_ix$ and $(S-s)V_ix$ are mainly supported near vertex $i$, and
\begin{equation}
(S-s)V_1x+(S-s)V_2x=(S-s)(S^2-S)x=0,
\end{equation}
and both terms must be small:
\begin{equation}
\norm{(S-s)V_ix}<e^{-r}.
\end{equation}
Therefore,
\begin{dmath}
	\norm{(Q_1-s)V_1x}\leq\norm{(S-s)V_ix}+\norm{(S-Q_1)V_1x}\\
	<e^{-r}\leq\frac{e^{-r}}{|s^2-s|}\norm{V_1x}\defeq\delta'||V_1x||.
\end{dmath}
According to Lemma \ref{lemma-approx}, $Q_1$ has an eigenvalue in $(s-2\delta',s+2\delta')$. Therefore, $s$ must be near (within $\sim\delta$) a window, and an argument similar to \eqref{eq-121} shows that $x$ cannot be a new eigenstate.

Due to this correspondence, the last term in the decomposition
\begin{dmath}
|\trv(S)-(\trv(Q_1)+\trv(Q_2))|\\
\leq|\sum_{s\in U_{\epsilon}}s|+|\sum_{q_1\in U_\epsilon}q_1|+|\sum_{q_2\in U_\epsilon}q_2|+|\sum_{s\notin U_{\epsilon}}s-\sum_{q_1\notin U_\epsilon} q_1-\sum_{q_2\notin U_\epsilon} q_2|,
\end{dmath}
is bounded by $\delta\ln^2\frac{1}{\epsilon}$ and goes to 0 as $r\to \infty$. The second and third term can be bounded due to \refeq{eq-EDPistraceclass}, since $Q_i-\overline{Q_i}$ obeys EDP with the same $C_1$ and $C_2$. For the first term, $S^2-S$ also obeys EDP (with respect to the central point 0), however with a new constant $C'_1\sim C_1e^{r}$. This is because (see \refeq{eq-S2-S}) the EDP of $V_i$ is with respect to vertex $i$, so the decay property of $S^2-S$ with respect to vertex 0 need to be estimated by $e^{-\dist{\vx,0}}<e^{r}e^{-\dist{\vx,1}}$. Luckily, similar to \refeq{eq-EDPistraceclass}, we still have

\begin{equation}\label{eq-later}
|\sum_{s\in U_{\epsilon}}s|<\int_{0}^{\epsilon} \frac{1}{C_2^2}\ln^2{\frac{ C_1e^{r}}{C_2^\alpha x}}dx\to 0,
\end{equation}
as $r\to\infty$ as long as $\epsilon=o(\frac{1}{r^2})$.
Thus we have proved that
\begin{equation}
\lim_{r\to\infty}\trv(S)-(\trv(Q_1)+\trv(Q_2))=0.
\end{equation}
Note that we have assumed that the requirement $\epsilon=o(\frac{1}{r^2})$ is compatible with $\delta=\frac{e^{-r}}{\epsilon}<\text{gap}$. This technical assumption is reasonable. Indeed, according to Lemma \ref{lemma-est}, the spectral gap at $\epsilon$ is roughly $(\frac{d}{d\epsilon}\ln^2\frac{1}{\epsilon})^{-1}\sim\frac{\epsilon}{\ln\epsilon}$ in average. In order for $\delta<\frac{\epsilon}{\ln\epsilon}$, it is enough to set $\epsilon>\Omega(e^{-C'r})$, which is exponentially smaller than $\frac{1}{r^2}$ for large $r$. Even if we consider the fluctuation of the spectral gaps and even if the level statistics is Poissonian (so that no level repulsion), the probability for this to be true is 1 from the following estimation:
\begin{dmath}
	\Pr(\text{at least one gap}<\frac{p(r)e^{-r}}{\epsilon})<\sum_{x>\epsilon}\frac{p(r)e^{-r}/\epsilon}{x/\ln x}<\frac{p(r)e^{-r}\ln\epsilon}{\epsilon^2} \times \ln^2\frac{1}{\epsilon}\to 0
\end{dmath}

\subsection{Proof of the finite size approximation}\label{sec-finiteproof}
In this section, we prove Proposition \ref{prop-finitesize}. The technique used will be similar to the above section. We need to compare vertex eigenvalues of $Q$ and $Q'_N$. Recall that $\norm{Q'_N-Q_N}\leq\rho^2,\norm{V'_N-W_N}<\rho,\norm{V'_N-W_N}<\rho$ where $\rho\sim e^{-Cr}$.

Temporarily fix $\epsilon$, and only consider eigenvalues outside $U_\epsilon\defeq(-\epsilon,\epsilon)\cup(1-\epsilon,1+\epsilon)$. For any $\alpha\notin U$, $|\alpha^2-\alpha|>\epsilon/2$. 

For $(Q-q)x=0$, $q\notin U_\epsilon$, we separate it as $x=y+z$ with respect to circle $r/2$, again $\norm{z}<\frac{p(r)e^{-Cr}}{\epsilon}\defeq\delta$. In the following $\delta$ means ``everything that goes like $\frac{p(r)e^{-Cr}}{\epsilon}$ with perhaps different $C$". Thus
\begin{equation}
\begin{aligned}
&\norm{(Q'_N-q)x}=\norm{(Q'_N-Q_N+Q_N-Q)x}\\
\leq&\norm{(Q'_N-Q_N)x}+\norm{(Q_N-Q)y}+\norm{(Q_N-Q)z}\lesssim \delta,
\end{aligned}
\end{equation}
so $Q'_N$ has an eigenvalue $q'\in(q-2\delta,q+2\delta)$ with eigenstate $x'$ satisfies $|(x,x')|>\sqrt{\frac{3}{4N}}$ (Lemma \ref{lemma-approx}). This implies $x'$ must contain a vertex eigenstate. Indeed, if not, we have $V'_Nx'=0$ so $x'=\frac{1}{{q'}^{2}-q'}W'_Nx'$ which is concentrated near boundary $r$, thus
\begin{dmath}
|(x,x')|=\frac{1}{|{q'}^2-q'|}(x,W'_Nx')\\
<\frac{2}{\epsilon}[|(y,W'_Nx')|+|(z,W'_Nx')|]\lesssim\delta,
\end{dmath}
a contradiction as $r\to\infty$.

On the other hand, if $(Q'_N-a)x=0 (a\neq 0,1)$, and $x$ is a vertex state: $W'_Nx=0$, then $x=\frac{V'_Nx}{a^2-a}$. We have
\begin{dmath}
(Q-a)x=\frac{1}{a^2-a}(Q-Q'_N)V'_Nx\\
=\frac{1}{a^2-a}[(Q-Q_N)V_Nx+(Q_N-Q'_N)V'_Nx]\lesssim\delta.
\end{dmath}
So, $Q$ has an eigenvalue in $(a-2\delta,a+2\delta)$.

Now we choose $r$ according to the same technical assumption above, so that there is a correspondence outside region $U_\epsilon$ for $\forall\epsilon$. Then, similarly we have
\begin{equation}
|\sum_{\text{vertex}}q'-\trv(Q)|\leq|\sum_{q\in U_{\epsilon}}q|+|\sum_{q'\in U_\epsilon}q'|+|\sum_{q\notin U_{\epsilon}}q-\sum_{q'\notin U_\epsilon}q'|.
\end{equation}
The last term is bounded by $\delta\ln^2\frac{1}{\epsilon}$ which goes to 0 as $\epsilon\to 0$. The first term also goes to 0 since $Q-\bQ$ is trace class. The second term is bounded by $\text{No.~}\{q'\}\epsilon$. Obviously $\text{No.~}\{q'\}<\dim Q_N\sim r^2$, so this term also converges to 0 since $\epsilon=o(\frac{1}{r^2})$.

\subsection{Proof that \texorpdfstring{$S-T$}{S-T} is trace class}\label{sec-sttraceclass}
In this section, we prove the claim used in Property \ref{prop-S0}.

The first step is to figure out the decay behaviour of the matrix elements of $S-T$.
According to the Peierls substitution \cite{Peierls},
\begin{equation}
S_{ij}=P_{ij}e^{i\int_i^j\vA\cdot d\vr},~~T_{ij}=P_{ij}e^{i\int_i^j\vA'\cdot d\vr},
\end{equation}
where $\vA$ and $\vA'$ are the vector potential for the two flux configurations. See Fig.~\ref{pic-geo}(a) (all angles here are directed), we have:
\begin{equation}
\begin{aligned}
\vA\cdot d\vr\propto 2\theta,~~\vA'\cdot d\vr\propto \theta_1+\theta_2,\\
(S-T)_{\vx,\vy}\sim P_{\vx,\vy}e^{i(\theta_1+\theta_2)}(e^{i(2\theta-\theta_1-\theta_2)}-1).
\end{aligned}
\end{equation}
\begin{figure}
 \centering
  \includegraphics[width=0.8\columnwidth]{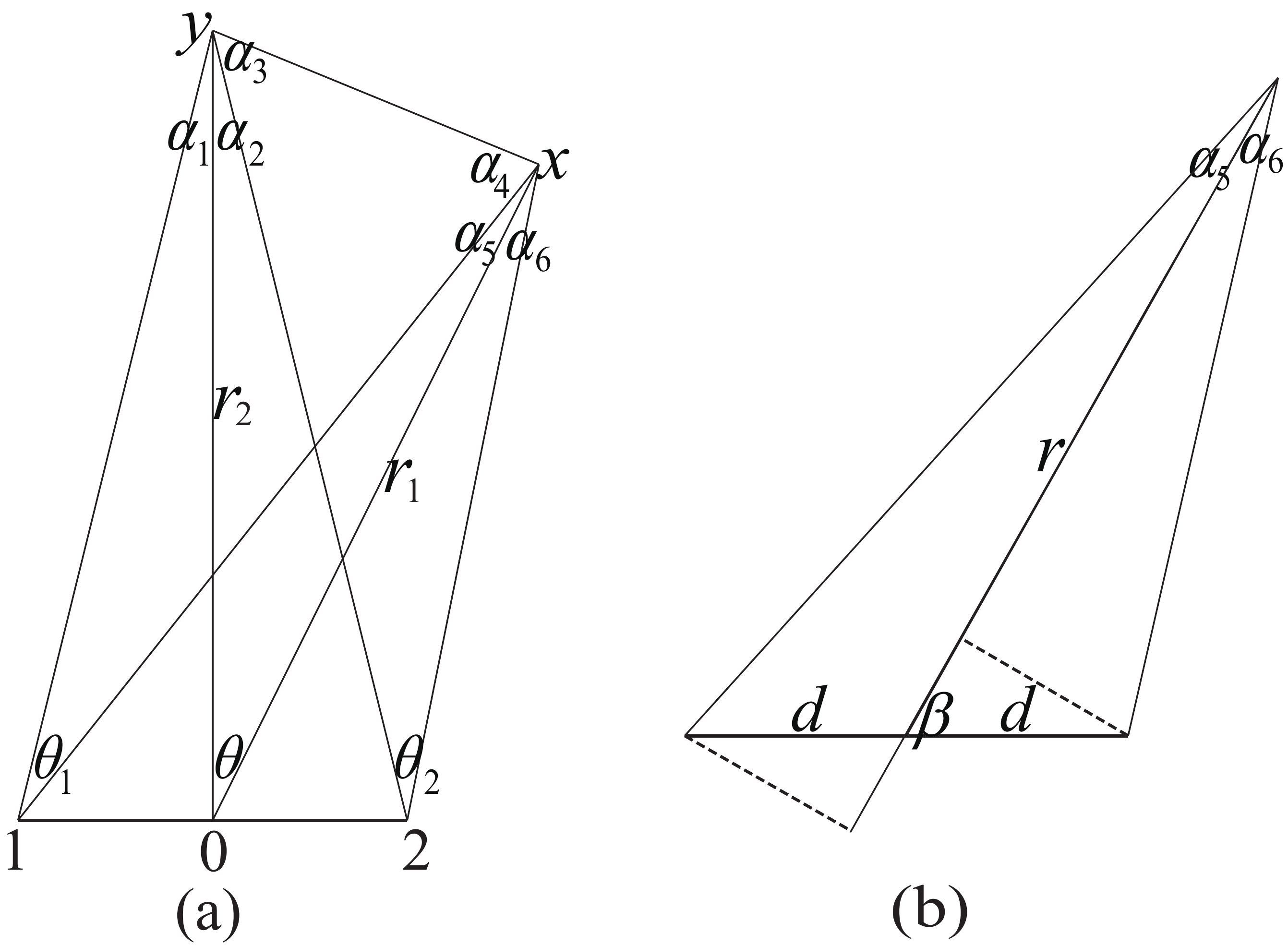}\\
 \caption{Relevant geometries in the proof. }\label{pic-geo}
\end{figure}

From geometry, $2\theta-\theta_1-\theta_2=(\alpha_1-\alpha_2)-(\alpha_5-\alpha_6)$. Let us calculate $(\alpha_5-\alpha_6)$. See Fig.~\ref{pic-geo}(b), we have:
\begin{equation} 
\begin{aligned}
\alpha_5-\alpha_6 & = \arctan\frac{d\sin\beta}{r+d\cos\beta}-\arctan\frac{d\sin\beta}{r-d\cos\beta} \\
 & = \arctan\frac{\frac{d\sin\beta}{r+d\cos\beta}-\frac{d\sin\beta}{r-d\cos\beta}}{1+\frac{d\sin\beta}{r+d\cos\beta}\frac{d\sin\beta}{r-d\cos\beta}}\\
&=-\arctan\frac{d^2\sin 2\beta}{r^2-d^2}\\
&=-\frac{d^2\sin 2\beta}{r^2}+\mathcal{O}(\frac{1}{r^4}).
\end{aligned}
\end{equation}
Due to the energy gap, $|A_{\vx,\vy}|\lesssim e^{-C|\vx-\vy|}$. Assuming $r_1\geq r_2$, we claim that
\begin{equation}
e^{-C|\vx-\vy|}|e^{i(2\theta-\theta_1-\theta_2)}-1|<e^{-\frac{C}{2}|\vx-\vy|}\mathcal{O}(\frac{1}{r_1^3}).
\end{equation} 
Indeed, if $e^{-\frac{C}{2}|\vx-\vy|}<\frac{1}{(r_1+r_2)^3}$, there is nothing needed to prove. If not, then $|\vx-\vy|<\frac{6}{C}\ln(r_1+r_2)<\frac{7}{C}\ln r_1$ (asymptotically). In this case, from geometry, we know $|\beta_i-\beta_j|=|\theta|\lesssim\frac{|\vx-\vy|}{r_1}$. Then $2\theta-\theta_1-\theta_2=(\alpha_1-\alpha_2)-(\alpha_5-\alpha_6)=\frac{\sin2\beta_1}{r_1^2}-\frac{\sin2\beta_2}{r_2^2}+\mathcal{O}(\frac{1}{r_2^4})$ will be of order $\mathcal{O}(\frac{|\vx-\vy|}{r_1^3})$ as can be seen from Taylor expansion. Then it is easy to see that the claim also holds.

The result is (in a more symmetric fashion, ignore constants) as follows: the operator $A= S-T$ satisfies the following decay property:
\begin{equation}\label{eq-decay}
|A_{\vx,\vy}|<\frac{1}{(|\vx|+|\vy|)^3}e^{-|\vx-\vy|}.
\end{equation}

Now, we prove this kind of operator must be trace class. Let us denote the $n^{\text{th}}$ singular value (decreasing order) to be $s_n$. According to the Courant min-max principle \cite{Lax}, we have
\begin{equation}
s_n=\min_{Y_{n-1}}\max_{u\perp Y_{n-1}}\frac{(Au,Au)}{(u,u)},
\end{equation}
where $Y_{n-1}$ means a subspace of dimension $n-1$. Thus, for any given n-dimensional subspace $Y_{n-1}$, we have
\begin{equation}
s_n^2\leq\max_{u\perp Y_{n-1}}\frac{(Au,Au)}{(u,u)}=\max_{\substack{u\perp Y_{n-1} \\||u||=1 }}||Au||^2.
\end{equation}
Let us choose the subspace $Y_{n-1}$ to be spanned by the $n$ components nearest to the center (so that the label of the components are approximately in the disk of radius $r\sim\sqrt{n}$). Denote the columns of $A$ to be $v_\vx$ ($v_\vx=Ae_\vx$, $\vx\in\mathbb{Z}^2$ is the label). With \refeq{eq-decay} it is easy to show (note that here $e^{-|\vx-\vy|}$ means $e^{-C|\vx-\vy|}$ for a different $C$)
\begin{equation}\label{eq-vxvy}
|(v_\vx,v_\vy)|=|(A^2)_{\vx,\vy}|\lesssim\frac{e^{-|\vx-\vy|}}{|\vx|^3|\vy|^3}.
\end{equation}
Thus,
\begin{dmath}
\norm{Au}^2=\norm{\sum_{|\vx|>r}u_\vx v_\vx}^2=\sum_{|\vx|,|\vy|>r}{\bar{u}_\vx u_\vy (v_\vx,v_\vy)}\\
=(\sum_{\substack{|\vx-\vy|\geq l\\|\vx|,|\vy|>r}}+\sum_{\substack{|\vx-\vy|<l\\|\vx|,|\vy|>r}})\bar{u}_\vx u_\vy (v_\vx,v_\vy).
\end{dmath}
Here, $l$ will be of the order $\ln r$, to be specific later.

The first summation is (crude but enough) controlled by $e^{-l}$ due to \refeq{eq-vxvy} and Cauchy inequality. For the second summation, we have
\begin{equation}
|\sum \bar{u}_\vx u_\vy (v_\vx,v_\vy)|<\frac{1}{4}\sum (|u_\vx|^2+|u_\vy|^2)|(v_\vx,v_\vy)|\lesssim \frac{l^2}{r^6}.
\end{equation}
Let us choose $l$ such that $e^{-l}=\frac{1}{r^6}$, we finally have
\begin{equation}
s_n^2\leq e^{-l}+\frac{l^2}{r^6}<\frac{\ln^2r}{r^6}\sim\frac{\ln^2n}{n^3}.
\end{equation}
so $\sum s_n=\sum \frac{\ln n}{n^{3/2}}$ converges, which means $A$ is trace class.


\bibliography{TIz2invariant}

\end{document}